\documentclass[aps,twocolumn,notitlepage,showpacs,superscriptaddress]{revtex4-1}

\usepackage{dcolumn}   
\usepackage{bm}        
\usepackage{amssymb}   
\usepackage[none]{hyphenat}

\usepackage[usenames,dvipsnames]{xcolor}

\usepackage[utf8]{inputenc}
\usepackage{amsmath,amsfonts,amssymb,amscd,amsthm,xspace}
\usepackage{slashed}
\usepackage{epsfig}
\usepackage{epstopdf}
\usepackage{mathrsfs}

\usepackage{dsfont}
\usepackage{graphicx}
\usepackage{xcolor}
\usepackage{hyperref}
\hypersetup{urlcolor=black, colorlinks=false}

\usepackage{tikz-cd}

\usepackage{ctable} 
\usepackage{multirow}
\usepackage{booktabs}
\usepackage{soul}

\newcommand\beq{\begin{equation}}
\newcommand\eeq{\end{equation}}
\newcommand\beqa{\begin{eqnarray}}
\newcommand\eeqa{\end{eqnarray}}

\newcommand{\nn}{\nonumber}

\newcommand{\la}{\langle}
\newcommand{\ra}{\rangle}
\newcommand{\id}{\mathrm{Id}}

\newcommand{\VV}{\mathcal{V}}
\newcommand{\II}{\mathcal{I}}
\newcommand{\TT}{\mathcal{T}}
\newcommand{\UU}{\mathcal{U}}
\newcommand{\HH}{\mathcal{H}}
\newcommand{\s}{\mathcal{S}_1}
\newcommand{\BB}{\mathcal{B}}
\newcommand{\ZZ}{\mathcal{Z}}
\newcommand{\PP}{\mathcal{P}}
\newcommand{\Tr}{\mathrm{Tr}}
\newcommand{\rmT}{\mathrm{T}}

\newcommand{\Ck}{\mathbb{C}^k}

\newcommand{\UPQP}{\mathrm{UPQP}_d}
\newcommand{\epsUPQP}{\varepsilon - \mathrm{UPQP}_d}

\newcolumntype{P}[1]{>{\centering\arraybackslash}p{#1}}
\newcolumntype{?}{!{\vrule width .5pt}}

\newtheorem{observation}{Observation}

\newtheorem{fact}{Fact}

\newtheorem{definition}{Definition}
\newtheorem{proposition}{Proposition}
\newtheorem{theorem}{Theorem}
\newtheorem{corollary}{Corollary}

\theoremstyle{definition}

\begin{document}

\title{Resource Quantification for the No-Programming Theorem}
\author{Aleksander M. Kubicki}
\email{amkubickif@gmail.com}
\affiliation{Departamento de Análisis Matemático, Universitat de València, 46100 Burjassot, Spain}
\author{Carlos Palazuelos}
\author{David Pérez-García}
\affiliation{Departamento de Análisis Matemático y Matemática Aplicada, Universidad Complutense de Madrid, 28040 Madrid, Spain}
\affiliation{Instituto de Ciencias Matemáticas, 28049 Madrid, Spain}

\date{\today}
\begin{abstract}
The no-programming theorem prohibits the existence of a Universal Programmable Quantum Processor. This statement has several implications in relation to quantum computation, but also to other tasks of quantum information processing, making this construction a central notion in this context. Nonetheless, it is  well known that even when the strict model is not implementable, it is possible to conceive of it in an approximate sense. Unfortunately,   the minimal resources necessary  for this aim are still not completely understood. Here, we investigate quantitative statements of the theorem, improving exponentially previous bounds on the resources required by such a hypothetical machine. The proofs exploit a new connection between quantum channels and embeddings between Banach spaces which allows us to use classical tools from geometric Banach space theory  in a clean and simple way.
\end{abstract}

\maketitle

Since the early days of Quantum Information Theory,  no-go theorems have served as  guideline in the search of a deeper understanding of quantum theory as well as for the development of applications of quantum mechanics to cryptography and computation. They shed light on those aspects of quantum information which make it so different from its classical counterpart. Some renowned examples are the no-cloning \cite{Herbert82,*Dieks82,*Milonni82,*WootersZurek82,*Mandel83}, no-deleting \cite{PatiBraunstein00} and no-programming \cite{NielsenChuang97} theorems. 

The no-programming theorem concerns with the so-called Universal Programmable Quantum Proccesor, UPQP \footnote{  Originally called Programmable Quantum Gate Array \cite{NielsenChuang97}.}. A UPQP is a universal machine  able to perform any quantum operation on an arbitrary input state of fixed size, programming the desired action in a quantum register inside the machine (a quantum memory). It can be understood as the quantum version of a stored-program computer. For the sake of simplicity, we will consider programmability of unitary operations, although this is not really a restrictive assumption \footnote{Even in this case we could program general quantum channels implementing a unitary first and tracing out a part of the output.}. 
With this figure of merit, the no-programming theorem is stated as the non-existence of a UPQP using finite dimensional resources. The key observation made in \cite{NielsenChuang97} is that in order to program two different unitaries we need two orthogonal program states. Then, the infinite cardinality of the set of unitary operators, even in the case of a qubit, leads immediately to the requirement of an infinite dimensional memory. Similar consequences follow for the related  concept of Universal Programmable Quantum Measurements \cite{DusekBuzek02,FiurasekDusekFilip02,DarianoPerinnoti05},  which are machines with the capability to be programmed to implement arbitrary quantum measurements.

From a conceptual point of view, the no-programming theorem points out  severe limitations in how universal quantum computation can be conceived. However, these limitations can be surpassed  by relaxing the requirements on the model of UPQP. In particular, one can consider programmable devices working noisily or probabilistically. Indeed, in the last two decades, several  proposals of such approximate UPQPs  have appeared in the literature \cite{NielsenChuang97,KimEtal01,*HilleryBuzekZiman01,*BrazierBuzek05,*BeraEtal09,*VidalCirac02,IshizakaHiroshima08}. Thus, it is interesting to look for more quantitative statements about \emph{quantum programmability}. To put it in explicit words, we worry here about the relation between the memory size of an approximate UPQP, $m$, and both, the accuracy of the scheme, $\varepsilon$,  and the size of the input register in which we want to implement the program, $d$. Despite their relevance, these relations are still poorly understood. Existing results are summarized  in table \ref{Table}.

\begin{table*}
	\begin{tabular}{c P{5cm} P{1cm} P{5cm} P{1cm}}	
		\rule{0pt}{2.5ex}				&		\multicolumn{2}{c}{   Previous results}   &	   \multicolumn{2}{c}{This work}
		\\ 
		\cmidrule[.5pt](l{2em}r{2em}){2-3} \cmidrule[.5pt](r{2em}l{2em}){4-5}	 
 \specialrule{.0em}{.15em}{.35em} 
		\multicolumn{1}{c?}{$\begin{array}{cc}\text{Lower}\\\text{bounds}\end{array}$}		&  	$ \begin{array}{cc} m \ge \mathsf{K}  (\frac{1}{d})^{\frac{d+1}{2}} \left(\frac{1}{\varepsilon} \right)^{\frac{d-1}{2}}   \\[.5em]   m \ge \mathsf{K} \left( \frac{d}{\varepsilon} \right)^2 \vspace{.1cm}  \end{array}$ & \multicolumn{1}{c?}{	$ \begin{array}{cc} \text{\cite{Perez06}}   \\[.5em] \text{\cite{Majenz17}} \end{array}$}	&	  $   m   \ge  2^{\frac{(1-\varepsilon) }{\mathsf{K}}d	-	\frac{2}{3}\log d}  $&[Th.\ref{th1}]
		\\ \specialrule{.5pt}{.0em}{.0em} 
		\multicolumn{1}{c?}{\rule[.2em]{0pt}{4ex}$\begin{array}{cc}\text{Upper}\\\text{bounds} \end{array}$}		&  	$ m  \le 2^{\frac{4 d^2 \log d}{\varepsilon^2} } $  &	  \multicolumn{1}{c?}{\cite{IshizakaHiroshima08,BeigiKonig11,Majenz18}} &	  $m  \le 
		\left(\frac{\mathsf{K}}{\varepsilon}\right)^{d^2} $ & [Eq.\eqref{Upperbound}]	\\
	\end{tabular}
	\caption{ \label{Table} Best known bounds for the optimal memory size of  UPQPs in comparison with the results presented here. Above, $\mathsf{K}$ denotes universal constants, not necessarily equal between them. Let us point out that the bound from \cite{Perez06} was deduced for programmable measurements instead of UPQPs. However, since a UPQP always can be turned into a Universal Programmable Quantum Measurement, this lower bound also applies for the case studied here. Notice that the alluded bound, although it enforces a strong scaling of $m$ with $\varepsilon$, becomes trivial for large input dimension $d$. It is in this regime where the bound from \cite{Majenz17} is more informative, but still exponentially weaker than the bound provided by Theorem \ref{th1}. }
\end{table*}

In this Letter we provide new upper and lower bounds which substantially clarify the ultimate resources required by approximate UPQPs.  Indeed, the results in this work entail exponential improvements over previously known results. Our bounds show the optimal dependence of $m$ with $\varepsilon$ and $d$ separately. In fact, the lower bound of Theorem \ref{th1} is nearly saturated for fixed $\varepsilon$ by the performance of Port Based Teleportation, which was originally conceived as a UPQP \cite{IshizakaHiroshima08}. On the other hand, we deduce an upper bound, \eqref{Upperbound}, saturating almost optimally the scaling with $\varepsilon$ of the bound from \cite{Perez06}.

The proofs presented in this manuscript are based on a connection with geometric functional analysis that we uncover. The use of techniques from this branch of functional analysis, in particular, from Banach space theory and operator spaces - as it is the case in this work - have proven to be very  fruitful in the study of different topics of quantum information theory such as entanglement theory, quantum non-locality and quantum channel theory (see \cite{AubrunSzarekBook,PalazuelosVidickReview} and references therein). We find the path to put forward this  mathematical technology to the framework studied here. More precisely, we \emph{characterize}  UPQPs as isometric embeddings between concrete Banach spaces which are in addition complete contractions (considering some operator space structure). Once this characterization is established, the results about UPQPs are deduced by using classical tools from local Banach space theory in a simple and clean way. We think that the general ideas presented  here and potential generalizations of them can provide further insights  in other contexts related with quantum computation and cryptography.

In this main text, we limit ourselves to explain the results stressing the ideas behind them \cite{SuppMat}.

\emph{Preliminaries}.-- Given a finite dimensional Hilbert space $\HH$, we denote by $\BB(\HH)$ or $\s(\HH)$ the space of (bounded) operators on $\HH$ with the operator or the trace norm, respectively. We will also denote by $\UU(\HH)$ and $\mathcal{D}(\HH)$  the subset of unitary and density operators. The set of quantum channels (that is, completely positive and trace
preserving maps $\BB(\HH)\rightarrow \BB(\HH) $) will be denoted by $\mathrm{CPTP} (\mathcal{H})$. We will usually consider a $d$-dimensional complex Hilbert space $\HH_d$ as the input state space and an ancillary $m$-dimensional complex Hilbert space $\HH_m$, as the memory of the programmable device under consideration.   When logarithms are used, they are generically considered in base 2.

\begin{definition}\label{PQP}
	 A quantum operation $\mathcal{P} \in \mathrm{CPTP}(\mathcal{H}_d \otimes \mathcal{H}_M)$ is a $d$-dimensional Universal Programmable Quantum Processor, $\UPQP$, if  for every $U \in \UU( \mathcal{H}_d)$ there exists a unit vector $ |\phi_U\rangle \in \HH_m$ such that:
	\begin{equation}\nn
	\mathrm{Tr}_{\mathcal{H}_m}  \left[ \mathcal{P} \left( \rho \otimes  |\phi_U\rangle \langle\phi_U|  \right) \right] = U \rho U^{\dagger},\quad  \text{for every \text{  }$\rho\in\mathcal{D}(\mathcal{H}_d $)}.
	\end{equation}
\end{definition}

Essentially, this is the concept of Universal Quantum Gate Array introduced in \cite{NielsenChuang97}, and whose impossibility is the content of the no-programming theorem discovered also there. As we said in the previous section, the no-programming theorem does not apply if one considers a relaxation of the previous definition; that is, in the case of approximate UPQPs. Two notions of approximate UPQPs have been considered in the literature: probabilistic settings \cite{NielsenChuang97,HilleryBuzekZiman02}, which implement exactly the desired unitary with some probability of failure, obtaining information about the success or failure of the procedure; and deterministic UPQPs \cite{VidalCirac00}, which always implement an operation which is close to the desired one. Notice that both notions are related, since probabilistic UPQPs can be also understood as deterministic ones just ignoring the information about the success or failure of the computation. A natural way to express these notions of approximation is  through the distance induced by the diamond norm \cite{Kitaev97}:

\begin{definition}\label{epsPQP}
	We say that $\mathcal{P} \in \mathrm{CPTP}(\mathcal{H}_d \otimes \mathcal{H}_m)$ is a $d$-dimensional ${\varepsilon}$--Universal Programmable Quantum processor, $\epsUPQP$, if  for every  $ U \in \UU( \mathcal{H}_d)$ there exists a unit vector $ |\phi_U\rangle\in\HH_m$ such that:
	\begin{equation}\nn
	\frac{1}{2}\left\|\mathrm{Tr}_{\mathcal{H}_m}  \left[ \mathcal{P} \left( \ \cdot \  \otimes  |\phi_U\rangle \langle\phi_U|  \right) \right] \   - \ U (\cdot) U^{\dagger}\right\|_{\diamond} \le \varepsilon,
	\end{equation}
	where $\|\: \cdot\: \|_\diamond$ denotes  the \emph{diamond norm}.
\end{definition}

Two  relevant examples of approximate UPQPs are,\begin{itemize}
	\item  the one built based on standard quantum Teleportation \cite{NielsenChuang97}, which can be understood as a probabilistic $\epsUPQP$  with $\varepsilon =1- \frac{1}{d^2}$ and memory dimension $m=d^2$;
	\item  the protocol of Port Based Teleportation itself \cite{IshizakaHiroshima08}, which can be arranged as a probabilistic or deterministic $\epsUPQP$ with memory dimension, $m$, scaling as $\sim 2^{\frac{4 d^2 \log d}{\varepsilon^2} }$ \cite{BeigiKonig11,Majenz18}.
\end{itemize}

Notice that in the first case, the resources used are remarkably efficient. The counterpart is that the success probability (accuracy of the setting) is rather low. In contrast, in the second example the accuracy can be arbitrarily improved at the  prize of increasing the dimension of the resource state. These examples show the rich landscape of behaviours displayed by UPQPs, which turns  the understanding of these objects challenging. The results presented here shed new light on them.

\emph{The connection}.-- In this section we explain, omitting proofs, the key connection between $\epsUPQP$ and isometric embeddings between Banach spaces at the heart of the proofs of our main results.

The crucial ingredient is the characterization of  $\UPQP$ as  isometric embeddings $\Phi: \s(\HH_d) \hookrightarrow \BB(\HH_m)$  with completely bounded norm $\|\Phi\|_{cb}\le1$, i.e., complete contractions. For $\epsUPQP$ the characterization holds distorting the isometric property of the embedding with some disturbance $\delta(\varepsilon)$. Turning to the completely bounded norm of $\Phi$, it can be understood in this particular case as follows. Let us put each $\VV\in\BB(\HH_d\otimes\HH_m)$ in one-to-one correspondence with the linear map
 \begin{equation}\label{correspondence}
\begin{array}{cccl}
	\Phi_\VV: & \s(\HH_d) &\lhook\joinrel\longrightarrow& \ \BB(\HH_m)\\
	&  \sigma &  \mapsto  &  \Phi_\VV(\sigma) := \Tr_{\HH_d} \VV (\sigma^\rmT \otimes \id_M)  .\end{array}
\end{equation} 
Given this correspondence, the completely bounded norm of $\Phi_{\VV}$ can be simply regarded as  $ \|\Phi_\VV\|_{cb} = \|\VV\|_{\BB(\HH_d\otimes \HH_m)}$.

More generally, within the theory of operator spaces, $\s(\HH_d)$ and $\BB(\HH_m)$ are endowed naturally with a sequence of norms when tensorized with the set of $k\times k$ complex matrices $M_k$. Then, $\|\Phi\|_{cb}$  is  nothing but  $\sup_{k} \| {\rm Id}_{M_k}\otimes \Phi\|$. The equivalence between this more profound definition of the completely bounded norm and the one given before is provided by a well known result in operator space theory \cite[Proposition 8.1.2]{RuanBook}. We emplace the interested reader to \cite{SuppMat}.

We are now in position to state formally the results of this section:
\begin{theorem}\label{prop1}
	Every unitary $\epsUPQP$, given by  $\mathcal {P}(\cdot)=\mathcal{V} (\,\cdot\,) \mathcal{V}^\dagger\in  \mathrm{CPTP} (\mathcal{H}_d \otimes \HH_m)$, defines a completely contractive map $\Phi_{\mathcal{V}}: S_1(\HH_d) \longrightarrow \BB(\HH_m)$ such that
	$$
	\nn\|\sigma\|_{S_1(\HH_d)} \ge \|\Phi_{\mathcal{V}}(\sigma)\|_{\BB(\HH_m)}  \ge (1-\varepsilon)^{1/2} \|\sigma\|_{S_1(\HH_d)}$$ 
	for every $\sigma\in S_1(\HH_d).$ Such a map  is called a completely contractive  $\varepsilon$-embedding.
\end{theorem}

We also found true a converse of this statement:
\begin{theorem}\label{prop2}
	Every completely contractive map $\Phi: \  \s(\HH_d) \longrightarrow \BB(\HH_m) $ such that
	$$\nn\|\sigma\|_{\s(\HH_d)} \ge \|\Phi(\sigma)\|_{\BB(\HH_m)}  \ge (1-\delta)\|\sigma\|_{\s(\HH_d)}$$
	 for every \,$\sigma\in S_1(\HH_d)$,	defines a $\epsUPQP$ with $\varepsilon=\sqrt{2\delta}  $ and memory dimension at most $d m^3 =\dim(\HH_d \otimes \HH_m^{\otimes 3})$.
\end{theorem}

This establishes a \emph{characterization}, rather than a simple relation between the objects considered.

Even when the proofs of these theorems were left outside the main text, let us finish this section noting that the starting point to establish this characterization is precisely the correspondence  considered above, \eqref{correspondence}.


\emph{Results about UPQPs}.-- The characterization given in the preceding section leads  to a better understanding of UPQPs, which is summarized in the last column of Table \ref{Table}. This results, presented in the remaining of the Letter,  reduce drastically the existing gaps between previous lower an upper bounds in the study of UPQPs. We now sketch the proof of them.

Let us begin with the upper bound,
 \begin{equation} \label{Upperbound}
m \le \bigg( \frac{\tilde C}{\varepsilon}  \bigg)^{d^2},
\end{equation}
 where $\tilde C$  is a constant. Although this bound follows easily from a $\varepsilon$-net argument,  we find instructive to follow the lines of the proof of Theorem \ref{prop2} in this simplified case.

Firstly, we think at the level of embeddings between Banach spaces and consider the following mapping:
\begin{equation}\label{EmbHB}
\begin{array}{clll}
	\Phi:	&	\s(\HH_d)  &\quad  \longrightarrow  \quad &  \ell_\infty\big( \mathsf{ball}(\BB(\HH_d)) \big)  \\
	&	\hspace{3mm}	\sigma & \quad \ 	\mapsto  \quad	&  \hspace{2mm}	\big( \Tr[A \sigma^\rmT] \big)_{A\in\mathsf{ball}(\BB(\HH_d))}		,
\end{array}
\end{equation}
where  $\mathsf{ball}(X)$ denotes the unit ball of a Banach space $X$ and, for a given set $\mathcal{X}$, $\ell_\infty(\mathcal X)$ denotes the space of bounded functions from $\mathcal X$ to $\mathbb C$ endowed with the supremum norm. Then, it is straightforward to see that this embedding is isometric. Indeed, noting that $\BB(\HH_d)$ is the Banach dual of $\s(\HH_d)$, the embedding considered is usually recognized as a standard consequence of the Hahn-Banach theorem \cite{Rudin}.

In addition, the fact that $\ell_\infty(\mathcal X)$ can be understood as a commutative $C^*$-algebra guarantees that the bounded and completely bounded norms of any map $\Phi:E\rightarrow\ell_\infty (\mathcal X)$ coincide \cite[ Proposition 2.2.6]{RuanBook}. This also allows us to drop out the awkward transposition in \eqref{EmbHB}.

In order to obtain a finite dimensional version of the   embedding \eqref{EmbHB}, we discretize the image  by means of a $\varepsilon$--net on  $\UU(\HH_d)$. That is, we consider a finite sequence 
$	
\lbrace U_i \rbrace_{i=1}^{|\II|} \subset  \UU(\HH_d)
$
such that for every $U\in \UU(\HH_d)$ there exists an index $i\in\II\ $ verifying
$ \| U-U_i\|_{\BB(\HH_d)} \le \varepsilon.$ 
Then, we define the embedding
\begin{equation}\nn
\begin{array}{clcccc}
	\tilde \Phi:	&	\s(\HH_d)  &  \longrightarrow   &  	\ell_\infty (\mathcal I)  &  \lhook\joinrel\longrightarrow   & \BB(\HH_\II) \\[0mm]
	&	\hspace{3mm}	\sigma & 	\mapsto 	&  \hspace{1mm}	\big( \Tr[U_i \sigma] \big)_{i=1}^{|\II|} 	& \mapsto &	  \sum_{i\in\II} \Tr[U_i \sigma] \,|i\ra\la i| ,
\end{array}
\end{equation}
 being $\HH_\II$ a complex Hilbert space of dimension  $|\II|$.

 Now, it is an easy exercise to see that
\begin{equation*}
\|\sigma\|_{\s(\HH_d)} \ge \|\tilde \Phi(\sigma)\|_{\BB(\HH_\II)}  \ge \left(1-\frac{\varepsilon^2}{2}\right)  \|\sigma\|_{\s(\HH_d)}, 
\end{equation*} 
for every $\sigma\in \s(\HH_d)$. Then, $\tilde \Phi$ is a particular instance of a map in the conditions of Theorem \ref{prop2}, but its very simple structure allows to get to the conclusion of the theorem very easily in this case, as we show now.

The embedding $\tilde \Phi$ suggests considering  the channel $\PP( \cdot  ) = \VV(\, \cdot \,)  \VV^\dagger$, with $\VV \in \UU(\HH_d\otimes \HH_\II)$ being the controlled unitary:
\begin{equation}\label{deltanetUPQP}
\VV = \sum_i U_i \otimes |i\ra\la i|,\nn
\end{equation}
where the register $\HH_\II$ plays the role of a memory. Then, according to Definition \ref{epsPQP}, let us compute the diamond distance of this channel (with a suitable memory state) with any unitary $U \in \UU(\HH_d)$.  Since the action of the considered channel on the input state is unitary, the problem reduces in this case to compute the usual trace distance
\begin{align*}
 & \min_{i\in\II}  \ \max_{|\psi\ra \in \mathsf{ball}(\HH_d)} \ \ \frac{1}{2} \| U_i |\psi\ra\la\psi| U_i^\dagger - U|\psi\ra\la\psi| U^\dagger\|_1 \\ & \hspace{6mm} =  \min_{i\in\II}  \  \max_{  |\psi\ra \in \mathsf{ball}(\HH_d)}  \sqrt{1-|\la \psi|U_i^\dagger U|\psi\ra|^2} \\
 	& \hspace{6mm} \le  \max_{ |\psi\ra \in \mathsf{ball}(\HH_d) }  \sqrt{1 - \left(1-\frac{\varepsilon^2}{2}\right)^2 }  \le \varepsilon.
\end{align*}

Therefore, the considered channel, $\PP$, is a $\epsUPQP$ with  memory dimension $|\II|$, the cardinality of the $\varepsilon$-net considered. This cardinal  can be taken lower than $  ( \tilde C /  \varepsilon )^{d^2}$ for some constant $\tilde C $ \cite[Theorem 5.11]{AubrunSzarekBook}, which is the announced bound.

 Due to the particular structure of the constructed $\epsUPQP$ we notice that the program states encoding different unitaries of the $\varepsilon$--net $\lbrace U_i \rbrace_{i=1}^{|\II|}$ are indeed orthogonal. This is in consonance with the fact discovered by Nielsen and Chuang  that for a $\UPQP$ ($\varepsilon=0$) any two program states encoding different unitaries  must be orthogonal \cite{NielsenChuang97}. Then, given a $\epsUPQP$, it is tempting to try to reverse the previous $\varepsilon$--net argument to find $|\II|$ mutually orthogonal program states, lower bounding in this way the dimension $m$ with the cardinality $|\II|$.  However,  in general ($\varepsilon>0$) the orthogonality between program states is no longer true (one can consider, for example, the case of Port Based Teleportation \cite{IshizakaHiroshima08}). Moreover, previous lower bounds  in \cite{Perez06} and \cite{Majenz17} (see Table \ref{Table}) were based precisely on direct $\varepsilon$--net arguments  
 {      which, in the end, essentially reduce to rough volume estimations.   It turns out that  the type constants (defined below) of the Banach spaces involved in Theorem \ref{prop1} give a more refined information of their geometry. This, together with the key property that type constants are preserved by subspaces, allows to conclude from Theorem \ref{prop1} the following exponential improvement over previous lower bounds on $m$:}

%

\begin{theorem} \label{th1}
	Let $\mathcal{P}\in \mathrm{CPTP}(\mathcal{H}_d \otimes \mathcal{H}_m)$ be a $\epsUPQP$, then $$\dim \HH_m \equiv m \ge 2^{ ^{\frac{(1-\varepsilon)}{3 C} d - \frac{2}{3} \log d}}$$ for some  constant $C$. Furthermore, if $\mathcal{P}$ is a unitary channel one has $m\ge	  2^{ \frac{(1-\varepsilon)}{C} d }$.
\end{theorem}

Let us sketch how Theorem \ref{th1} can be obtained from Theorem \ref{prop1}. For simplicity, we restrict  to the case where the considered UPQP  is a unitary channel. The general case can be handled by means of a Stinesprings dilation of the channel under consideration. 

The basic idea consists in studying $\varepsilon$-embeddings between  $\s(\HH_d)$ and $\BB(\HH_m)$. These two spaces are extremely far apart from each other as Banach spaces and it is this intuition which leads us to Theorem \ref{th1}.   A quick argument to study necessary conditions on the dimensions of the spaces involved is provided considering their type-2 constants. Given a Banach space $X$, its type-2 constant, $\mathrm{T}_2(X)$, is the infimum of the constants $\rmT$ satisfying the inequality
 	\begin{equation} \nn \hspace{-3mm}
  \left( \mathbb{E}   \Big[ \big\|  \sum_i \varepsilon_i  x_i \big\|_X^2 \Big]   \right)^{1/2}
  \hspace{-1mm} \le	 \rmT \left( \sum_{i=1}^n \|x_i\|_X^2 \right)^{1/2},
  \end{equation}
 for every  sequence $\lbrace x_i \rbrace_{i=1}^n \subset X$. There,  $\mathbb{E}[\, \cdot\, ]$ is the expected value over any combination of signs $\lbrace{\varepsilon_i}\rbrace_{i=1}^n \in \lbrace-1,1\rbrace^n$ with uniform weight $1/2^n$, that is, independent identically distributed Rademacher random variables \cite{SuppMat}. Let us point out that despite the great impact of the notion of  type/cotype on Banach space theory in the last decades, in the context of quantum information it only appeared very recently in \cite{BrandaoHarrow15}.
 
  Since  $\Phi_{\VV}$ in Theorem \ref{prop1}  maps $\s(\HH_d)$ into a subspace of $\BB(\HH_m)$ --with distortion $(1-\varepsilon)^{1/2}$-- the following relation between type constants of these spaces is enforced: 
 {\small\begin{align}\nn
 \mathrm{T}_2(S_1(\HH_d))  \le \frac{1}{(1-\varepsilon)^{\frac{1}{2}}} \mathrm{T}_2( \Phi_{\VV}(S_1(\HH_d)) )  \le \frac{1}{(1-\varepsilon)^{\frac{1}{2}}} \mathrm{T}_2(\BB(\HH_m)).
 \end{align}}
 The first inequality follows from $\Phi_\VV$ being a $\varepsilon$-embedding (in the sense of Theorem \ref{prop1}), while the second inequality follows from the property of type constants being preserved by subspaces. Introducing in those inequalities the following known estimates for type constants of the spaces involved:
\begin{align}\nn
\sqrt{d} \le \mathrm{T}_2(S_1(\HH_d)) ,\quad \mathrm{T}_2(\BB(\HH_m)) \le  \sqrt{C \log m},
\end{align} we obtain the desired bound:
\begin{equation}\nn
d  \le \frac{C}{(1-\varepsilon)}\log m \quad  \Rightarrow \quad  m \ge 2^{\frac{(1-\varepsilon)}{C} d}.\nn
\end{equation}

The constant here, as well as in the general case of nonunitary channels, can be taken equal to $4$.

To finish, let us mention that the type-argument sketched above can be made more explicit, obtaining bounds for the memory size necessary to program specific families of unitaries \cite{SuppMat}.

 \emph{Discussion}.-- In this work we have studied the minimal conditions, in terms of resources, that have to be satisfied by approximate UPQPs.
 The bounds presented here have clarified several questions about optimality of this conceptual construction. In fact, we have almost closed the gaps in the optimal scaling of the memory size of UPQPs with the accuracy $\varepsilon$  and input dimension $d$, when considered separately. 
 
 Firstly, we have deduced the upper bound \eqref{Upperbound} giving a construction based on a $\varepsilon$--net on $\UU(\HH_d)$. In this sense,
 this construction can be seen as a generalization to the case of UPQPs of the programmable measurement introduced in \cite{DarianoPerinnoti05}. As in that case, our proposal improves exponentially the memory resources consumed by other known constructions (see Table \ref{Table}). In fact, the bound \eqref{Upperbound}  exponentially improves the scaling with the accuracy $\varepsilon$ of Port Based Teleportation and nearly saturates  the lower bound deduced in \cite{Perez06} in the context of Universal Programmable \emph{Measurements}. This shows that, indeed, this is the optimal dependence on that parameter also in the case of UPQPs. More generally, it also outperforms Port Based Teleportation whenever $\tilde C / \varepsilon \le d^{4/\varepsilon^2}$. Obviously, the drawback is that the optimal $\epsUPQP$ constructed here cannot be used to achieve any kind of teleportation.

 On the other hand, the main result obtained is the lower bound expressed by Theorem \ref{th1}. The first and most obvious consequence of this result is that for any fixed value of $\varepsilon$, the dimension of the memory of a $\epsUPQP$ must scale exponentially in $d$. Indeed, in this case the dependence with $d$ in the stated lower bound is exponentially stronger than all known previous results. Furthermore, this bound is saturated in this sense by the performance of Port Based Teleportation, referred in  table \ref{Table} as the best upper bound for $m$.

Notwithstanding, more difficult relations $\varepsilon$--$d$ can be considered, being the general scaling in this case still an open question. However, we also contribute to this point giving an upper bound for the achievable accuracy by UPQPs with memory of size $\mathsf{poly}(d)$. As a straightforward consequence of Theorem \ref{th1} we obtain the next:
\begin{corollary}
	For any $\epsUPQP$ with memory size $m \le k d^s$ for some constants $k$, $s$, the following inequality is satisfied:
	$$\varepsilon \ge 1 - C'_{k,s} \frac{\log d }{d}, $$ 
	where $C'_{k,s} = 3 C(s+\log k+ 2/3) $. 
\end{corollary}
This severely restricts the  accuracy achievable by $\epsUPQP$ with polynomially sized memories.

Moreover, due to the relation between UPQPs and other tasks as quantum teleportation \cite{NielsenChuang97,IshizakaHiroshima08}, state discrimination \cite{DusekBuzek02,SedlakZiman07,*SentisBagan11,*ZhouXinWu12,*Zhou12,*SentisBagan13},  parameter estimation \cite{JiWangDuanFengYing08,*Suzuki16}, secret and blind computation \cite{GiovannettiMaccone13,*PerezFitzsimons15}, homomorphic encryption \cite{PerezFitzsimons14}, quantum learning of unitary transformations \cite{BisioDarianoEtal10}, etc., we believe that the knowledge about them could also be relevant in a wide variety of topics. For example, as a direct application of the results presented here, we also obtain a lower bound for the dimension of the resource space necessary to implement deterministic Port Based Teleportation. There exist more accurate bounds for this particular case, see  \cite{Majenz18}, but notice that we did not use in any way the many symmetries presented in that protocol, and our bound is generic for any protocol implementing, in some sense, a UPQP. Furthermore, it is deduced from our results that the unavoidable exponential scaling with $\varepsilon^{-1}$ in the case of Port Based Teleportation, comes entirely from the signalling restrictions imposed  in this protocol, and cannot be deduced from the programming properties of it.

Finally, some interesting questions related with the work presented here arise. The most direct one is whether it is possible to deduce a lower bound on $m$ unifying the bound from \cite{Perez06} and the bound from Theorem \ref{th1}. This could give more information about optimality of UPQPs in cases beyond the scope of this work. In relation with that, it would be desirable to improve the exponents in the bounds to match exactly lower and upper bounds, though this will not affect qualitatively the consequences presented here. Further on, it would be also very interesting to look for relations between memory requirements on UPQPs and circuit complexity problems. A way to explore this line could consist on looking for  correspondences between circuits and memory states in  UPQPs. 
 \vspace{5mm}

 We thank K. Chakraborty and C. Majenz for valuable discussions during the preparation of this work. We acknowledge financial support from Spanish MICINN (grants MTM2014-54240-P, MTM2017-88385-P and Severo Ochoa project SEV-2015-556), Comunidad de Madrid (grant QUITEMAD+CM, ref. S2013/ICE-2801). A.M.K. is supported by the Spanish MICINN project MTM2014-57838-C2-2-P and C.P. is partially supported by the ``Ramón y Cajal program" (RYC-2012-10449).  This project has received funding from the European Research Council (ERC) under the European Union’s Horizon 2020 research and innovation programme (grant agreement No 648913).


%

\cleardoublepage

\onecolumngrid

\setcounter{definition}{0}
\setcounter{theorem}{0}
\setcounter{proposition}{0}
\setcounter{corollary}{0}

\begin{center}\large \textbf{ Supplemental Material } \end{center}
\hrule
\vspace{1cm}

We concentrate the technical part of the paper  here, where we provide detailed proofs of the statements made in the main text. For the reader's convenience, we recall here the main results in the Letter:

\begin{theorem}\label{props1}
	Every unitary $\epsUPQP$, given by  $\mathcal {P}(\cdot)=\mathcal{V} (\,\cdot\,) \mathcal{V}^\dagger\in  \mathrm{CPTP} (\mathcal{H} \otimes \HH_m)$, defines a completely contractive map $\Phi_{\mathcal{V}}: S_1(\HH_d) \longrightarrow \BB(\HH_m)$ such that
	$$
	\nn\|\sigma\|_{S_1(\HH_d)} \ge \|\Phi_{\mathcal{V}}(\sigma)\|_{\BB(\HH_m)}  \ge (1-\varepsilon)^{1/2} \|\sigma\|_{S_1(\HH_d)}$$ 
	for every $\sigma\in S_1(\HH_d).$
\end{theorem}

\begin{theorem}\label{props2}
	Every completely contractive map $\Phi: \  \s(\HH_d) \longrightarrow \BB(\HH_m) $ such that
	$$\nn\|\sigma\|_{\s(\HH_d)} \ge \|\Phi(\sigma)\|_{\BB(\HH_m)}  \ge (1-\delta)\|\sigma\|_{\s(\HH_d)}$$
	for every \,$\sigma\in S_1(\HH_d)$,	defines a $\epsUPQP$ with $\varepsilon=\sqrt{2\delta}  $ and memory dimension at most $d m^3 =\dim(\HH_d \otimes \HH_m^{\otimes 3})$.
\end{theorem}

\begin{theorem} \label{ths1}
	Let $\mathcal{P}\in \mathrm{CPTP}(\mathcal{H} \otimes \mathcal{H}_M)$ be an $\epsUPQP$, then $$\dim \HH_m \equiv m \ge 2^{ ^{\frac{(1-\varepsilon)}{3 C} d - \frac{2}{3} \log d}}$$ for some  constant $C$. Moreover, if $\mathcal{P}$ is a unitary channel one has $m\ge	  2^{ \frac{(1-\varepsilon)}{C} d }$.
\end{theorem}

In order to present  complete and self-contained proofs of the statements above, we first need some additional definitions.

\section{More definitions}\label{SecProofs0}

\subsection{The diamond norm}
We begin reviewing the diamond norm \cite{Kitaevv97}, which is the norm providing the distance notion in this work:
\begin{definition}
	For any linear map $\PP: \BB(\HH_d)\rightarrow \BB(\HH_d)$ we define its diamond norm as:
	$$
	\| \PP \|_\diamond = \sup_{\substack{k\ge1 \\  \sigma \in \mathsf{ball}\big(\s(\HH_d \otimes \mathbb{C}^k)\big)} } \| \PP \otimes \id_{\mathbb{C}^k} (\sigma) \|_1,
	$$ 
	where $\|\cdot\|_1$ denotes the trace norm and $\mathbb{C}^k$ is the complex k-dimensional Hilbert space.
\end{definition}


The relevance of the distance induced by this norm in quantum information theory can be explained by its operational interpretation (see \cite[ Ch. 3, Theorem 3.52]{WatrousNotes}):
\begin{theorem}\label{Thdiamond}
	Let $\mathcal{P}_0,\ \mathcal{P}_1 \in \mathrm{CPTP}(\mathcal{H})$. Then,
	\beq
	p^*_{dist} = \frac{1}{2} + \frac{1}{4}\left\| \mathcal{P}_0 -  \mathcal{P}_1 \right\|_\diamond,
	\eeq
	where $p^*_{dist}$ is the optimal probability of distinguishing between the channels $\PP_0$ and $\PP_1$ when a single instance of one of them is provided with probability $\frac{1}{2}$ each.
\end{theorem}

This optimal probability $p^*_{dist} $ can be written explicitly as:
\beq\label{DistProb}
p^*_{dist} = \mathrm{sup} \bigg[ \frac{1}{2} \Tr \big[Q \big( \PP_0 \otimes \id_\mathcal{Z}  (\rho) \big) \big]  + \frac{1}{2}  \Tr \big[ \big(\id_{\HH_d\otimes\mathcal{Z}}-Q\big) \big(\mathcal{P}_1 \otimes \id_\mathcal{Z}  (\rho) \big) \big]  \bigg],
\eeq
where the supremum is taken over all finite dimensional Hilbert spaces $\Ck$, density operators $\rho \in \mathcal{D}(\HH_d \otimes \Ck)$ and binary POVMs $\lbrace Q,\id-Q\rbrace$. Moreover, the supremum is achieved for $\Ck\simeq\HH_d$ and $\rho$  pure. 

Another important notion in the characterization we give of UPQPs as $\varepsilon$-embedding is the completely bounded norm of a map $\Phi:\s(\HH_d)\rightarrow \BB(\HH_m)$. The most natural way to introduce this notion is by understanding the completely bounded maps (maps with finite completely bounded norm) as the natural morphisms in the category of Operator Spaces. Then, let us briefly recall the key notions about these objects.

\subsection{The completely bounded norm and operator spaces}
An operator space is a complex Banach space $E$ together with a sequence of ``reasonable'' norms in the spaces $M_k \otimes E = M_k (E)$, where $M_k(E)$ is the space of square matrices of order $k$ with entries in $E$. This turns out to be equivalent to consider $E$ as a closed subspace of some $\BB(\HH)$ via a chosen embedding, defining the norm in $M_k(E)$ as the norm inherited from the embedding, $M_k(E) \subset M_k (\BB(\HH)) \simeq  \BB(\mathbb{C}^k\otimes \HH)$. As pointed out before, the natural morphisms compatible with this additional structure are the completely bounded (c.b.) maps. Given a linear map between operator spaces $\Phi:\, E \rightarrow F $, we define its completely bounded norm as $\|\Phi\|_{cb} := \sup_k \| \id_{M_k} \otimes \Phi:\, M_k(E)\rightarrow M_k(F)\|  $. Thus, the c.b. maps are those for which $\|\Phi\|_{cb}<\infty$, and we denote them by $\mathcal{CB}(E,F)$. Additionally, we say that a map is completely contractive (c.c.) if $\|\Phi\|_{cb} \le1$. Finally, c.b. maps provide us also with the notion of duality of an operator space $E$, $E^*$. For any $k\in\mathbb{N}$ we just define $M_k(E^*) = \mathcal{CB}(E,M_k)$.

Now, $\s(\HH)$ and $\BB(\HH)$ can be endowed with an operator space structure  as follows. As depicted in the previous paragraph, there is a natural operator space structure on $\BB(\HH)$ given by the identification $M_k(\BB(\HH)) = \BB(\mathbb{C}^k \otimes \HH)$. Then, $\s(\HH)$ inherits its operator space structure from the former by the duality explained in the previous paragraph.

\subsection{Type and cotype of a Banach space}
Next, we introduce the notion of type and cotype of a Banach space. In order to make this work completely accesible for non-experts we will only introduce the basic concepts we need to develop our results. A complete study about the theory of type/cotype of Banach spaces can be found in classical texts as \cite{TomczakJaegBook}.

For us, it will be enough to consider the Rademacher notion of type and cotype, built on Rademacher random variables, that is, random variables which takes values $-1$ and $1$ with probability $1/2$. Let us denote in this section by $\lbrace \varepsilon_i \rbrace_{i=1}^n$ a set of $n$ i.i.d. such random variables. Then, $\mathbb{E}[f(\lbrace \varepsilon_i \rbrace_{i=1}^n)]$ will be the expected value of the function $f$ over any combination of signs $\lbrace{\varepsilon_i}\rbrace_{i=1}^n \in \lbrace-1,1\rbrace^n$ with uniform weight $1/2^n$.

\begin{definition}\label{typedef}Let $X$ be a Banach space and let $1 \le p \le 2$. We say $X$ is of (Rademacher) type $p$ if there exists a positive constant $\rmT$ such that for every natural number $n$ and every sequence $\lbrace x_i \rbrace_{i=1}^n \subset X$ we have
	\beq \nn \hspace{-3mm}
	\left( \mathbb{E}   \Big[ \big\|  \sum_{i=1}^n \varepsilon_i  x_i \big\|_X^2 \Big]   \right)^{1/2}
	\hspace{-1mm} \le	 \rmT \left( \sum_{i=1}^n \|x_i\|_X^p \right)^{1/p},
	\eeq
	Moreover, we define the Rademacher type $p$ constant $\mathrm{T}_p(X)$ as the infimum of the constants $\rmT$ fulfilling the previous inequality.\end{definition}

Although we will not use it in this work, for the sake of completeness we will also introduce, for a given Banach space $X$ and $2\leq q < \infty$, the Rademacher cotype $q$ constant  $\mathrm{C}_q(X)$ as the infimum of the constant $\mathrm C$ (in case they exist) such that the following inequality holds for every natural number $n$ and every sequence $\lbrace x_i \rbrace_{i=1}^n \subset X$,
\begin{align*}\mathrm{\mathrm C}^{-1}\Big( \sum_{i=1}^n \|x_i\|_X^q\Big)^{1/q} \le 
	\left( \mathbb{E}   \Big[ \big\|  \sum_{i=1}^n \varepsilon_i  x_i \big\|_X^2 \Big]   \right)^{1/2} 
\end{align*}


These are central concepts in the theory of Banach spaces, context in which  the systematic use of these notions is traced back to the  work of  J. Hoffmann-J{\o}rggensen, S. Kwapie{\'n}, B. Maurey and G. Pisier in the 1970's. 

The following proposition is straightforward from the definition of $\mathrm{T}_p(X)$.
\begin{proposition}\label{TypeProp1}$\mathrm{T}_p(X)$  is preserved by subspaces. That is, if $S$ is a subspace of $X$ (as Banach spaces), then $\mathrm{T}_p(S) \le \mathrm{T}_p(X)$.
\end{proposition}

Although the proof of the following result is also very easy, we include it here for the sake of completeness.
\begin{proposition}\label{TypeProp2}
	Given a linear isomorphism between two Banach spaces $X$ and $Y$, $\Phi: X \rightarrow Y$, the following relation between type constants holds:
	\beq\label{Reltypeconsts}
	\rmT_p(X) \le \|\Phi\| \|\Phi^{-1}\| \rmT_p(Y).
	\eeq 
\end{proposition}
\begin{proof}
	Let us assume that $Y$ has type $p$ constant $\mathrm{T}_p(Y)$. Then, for any $n$ and any family $\lbrace x_i \rbrace_{i=1}^n \subset X$ we note that, since $\Phi$ is an isomophism, for any $i$: $x_i=\Phi^{-1}(y_i)$ for some $y_i\in Y$. Then, 
	\begin{align*}\nn
		&\left( \mathbb{E}   \Big[ \big\|  \sum_{i=1}^n \varepsilon_i  x_i \big\|_X^2 \Big]   \right)^{1/2}
		=\left( \mathbb{E}   \Big[ \big\|  \sum_{i=1}^n \varepsilon_i  \Phi^{-1}(y_i) \big\|_X^2 \Big]   \right)^{1/2} 
		\leq  \|\Phi^{-1}\|\left( \mathbb{E}   \Big[ \big\|  \sum_{i=1}^n \varepsilon_i  y_i \big\|_X^2 \Big]   \right)^{1/2} \\
		&\le \|\Phi^{-1}\| \mathrm{T}_p(Y) \Big( \sum_{i=1}^n \|y_i\|_Y^p\Big)^{1/p}
		=\|\Phi^{-1}\| \mathrm{T}_p(Y) \Big( \sum_{i=1}^n \|\Phi(x_i)\|_Y^p\Big)^{1/p} \le \|\Phi\| \|\Phi^{-1}\|    \mathrm{T}_p(Y) \Big( \sum_{i=1}^n \|x_i\|_X^p\Big)^{1/p}.
	\end{align*}
	Since $\rmT_p(X)$ is by definition the smallest constant satisfying the inequality above, the stated inequality must hold and we conclude our proof.
\end{proof}

Finally, we will show how the estimates used in the main text:
\beq\label{typeconsts}
(\dim \HH)^{1/2} \le \mathrm{T}_2(S_1(\HH)) ,\quad \mathrm{T}_2(\BB(\HH)) \le  \big(C \log (\dim\HH)\big)^{1/2}, \text{   }\text{ where }C\text{ is a constant},
\eeq
can be obtained.
\begin{proof}
	
	In order to show the first estimate, let us consider the elements (diagonal matrices) $x_i=|i\rangle\langle i|$ in $\s(\HH)$  for every $i=1,\cdots, \dim\HH$. Then, the bound can be easily obtained by considering the family $\lbrace x_i \rbrace_{i=1}^{\dim\HH}$ in Definition \ref{typedef}. 
	
	For the second estimate we begin recalling that for any $p \ge 2$, the Banach-Mazur distance between  the $p$-th Schatten class of operators on $\HH$, $\mathcal{S}_p(\HH)$, and $\mathcal{S}_\infty(\HH) \simeq \BB(\HH)$ is equal to $(\dim\HH)^{1/p}$\cite[Theorem 45.2]{TomczakJaegBook}. In particular, there exists a linear isomorphism $\Phi: \BB(\HH)\rightarrow\mathcal{S}_p(\HH)$ such that $\|\Phi\| \|\Phi^{-1}\| = (\dim \HH)^{1/p}$. In addition, it is also known that for $2\le p < \infty$, $ \rmT_2(\mathcal{S}_p(\HH)) \le \sqrt{p}$ \cite{Tomczak74}. Therefore, according to Proposition \ref{TypeProp2}, for every $2\leq p<\infty$ we obtain that
	$$
	\rmT_2(\BB(\HH)) \le (\dim \HH)^{1/p} \rmT_2(\mathcal{S}_{p}(\HH)) \le   (\dim \HH)^{1/p} \sqrt{p}.
	$$
	Considering $p=  \log_{\sqrt2} (\dim\HH) $  we obtain the second relation in \eqref{typeconsts}, with $C = 4$.  
\end{proof}


\section{Proofs}
\subsection{Theorem  \ref{props1}}\label{SecProofs1}

In this section we will show that a unitary quantum channel $\PP(\cdot) = \VV ( \cdot ) \VV^\dagger \in \mathrm{CPTP}(\HH_d\otimes \HH_m)$ implementing a $\epsUPQP$ induces a complete contraction $\Phi_\VV :\ \s(\HH_d) \hookrightarrow \BB(\HH_m)$ such that
\begin{align*}
	\nn\|\sigma\|_{S_1(\HH_d)} \ge \|\Phi_{\mathcal{V}}(\sigma)\|_{\BB(\HH_m)}  \ge (1-\varepsilon)^{1/2} \|\sigma\|_{S_1(\HH_d)}
\end{align*}for every $\sigma\in S_1(\HH_d)$.

\begin{proof}
	Given a unitary channel $\PP(\cdot) = \VV ( \cdot ) \VV^\dagger$, we consider the map 	$\Phi_\VV: \  \s(\HH_d) \longrightarrow \BB(\HH_m)$ defined as $$ \Phi_\VV(\cdot) := \Tr_{\HH_d} \VV (\:\cdot^{\rmT} \: \otimes \id_{\HH_m}),$$where $\VV\in  \BB(\HH_d\otimes\HH_m)$ is the unitary corresponding to $\PP$. As commented in the main text, in this case the completely bounded norm of $\Phi_\VV: \  \s(\HH_d) \longrightarrow \BB(\HH_m)$ coincides with $\|\VV\|_{\BB(\HH_d\otimes\HH_m)}=1$.\footnote{ Recall the completely isometric character of the correspondence $\BB(\HH_d\otimes\HH_m) \simeq \mathcal{CB}\left(\s(\HH_d),\BB(\HH_m)\right)$} Thus, $\Phi_\VV$ is completely contractive. In addition, since $$\|\Phi_\VV: \  \s(\HH_d) \rightarrow \BB(\HH_m)\|\leq \|\Phi_\VV: \  \s(\HH_d) \rightarrow \BB(\HH_m)\|_{cb}=1,$$ we immediately deduce that $ \|\Phi_{\mathcal{V}}(\sigma)\|_{\BB(\HH_m)}\leq \nn\|\sigma\|_{S_1(\HH_d)}$ for every $\sigma\in S_1(\HH_d)$.
	
	
	For the second inequality in \eqref{ineqs}, we elaborate on the norm:
	\beq\nn
	\|\Phi_{\VV}(\sigma)\|_{\BB(\HH_m)}  =  \sup   \| \mathrm{Tr}_{\HH_d}\left[ \VV (\sigma^\rmT \otimes \id_{\HH_m})\right] |\xi\ra\|_{\HH_m},
	\eeq
	where the supremum is taken over unit vectors $|\xi\ra \in \HH_m$. Now, we consider the singular value decomposition of  $\sigma^\rmT = \sum_i \mu_i |\psi_i\ra\la\gamma_i|$ with $(|\psi_i\ra)_{i=1}^d$, $(|\gamma_i\ra)_{i=1}^d$ orthonormal bases of $\HH_d$. Therefore $|\gamma_i\ra = U |\psi_i\ra$ for some unitary $U$. Furthermore, we can take $\mu_i\ge 0$ and then $\sum_i \mu_i = \|\sigma^\rmT \|_{S_1(\HH_d)} =\|\sigma\|_{S_1(\HH_d)}$, which can be restricted to one and the general case will follow by homogeneity. Besides, it is convenient to express $\sigma^\rmT$ as
	\begin{align}\nn
		\sigma^\rmT = \sum_i \mu_i |\psi_i\ra\la\psi_i| U^\dagger = \Tr_\ZZ	\big(\sum_i \sqrt{\mu_i}  |i\ra_\ZZ |\psi_i\ra \big)	\big(\sum_j  \sqrt{\mu_j} \la j |_\ZZ  \la \psi_j | U^\dagger \big)\equiv  \Tr_\ZZ  |\psi\ra\la\gamma |,
	\end{align}
	where we have considered a new auxiliary Hilbert space $\ZZ$ and $|\gamma\ra := (U\otimes\id_\ZZ)|\psi\ra$. Now, we are ready to rewrite
	\begin{align} 
		\|\Phi_{\VV}(\sigma)\|_{\BB(\HH_m)}  &=  \sup_{|\xi\ra \in \mathsf{ball} (\HH_m)}   \big\| \mathrm{Tr}_{\HH_d\otimes\ZZ}\left[( \VV \otimes \id_\ZZ) (|\psi\ra\la\gamma | \otimes \id_{\HH_m})\right] |\xi\ra \big\|_{\HH_m} \nn\\
		&=\sup_{|\xi\ra  \in \mathsf{ball} (\HH_m)} 
		\Big[ \Tr \big[ (|\gamma\ra \la \gamma| \otimes \id_{\HH_m})    (\VV \otimes \id_\ZZ )( |\psi\ra\la\psi|\otimes  |\xi\ra\la\xi| )  (\VV^\dagger \otimes \id_\ZZ )  \big]    \Big]^{\frac{1}{2}}
		\nn\\
		&\ge 
		\Big[ \Tr \big[ (|\gamma\ra \la \gamma| \otimes \id_{\HH_m})   (\VV \otimes \id_\ZZ )( |\psi\ra\la\psi|\otimes  |\xi_U\ra\la\xi_U| )  (\VV^\dagger \otimes \id_\ZZ )  \big]    \Big]^{\frac{1}{2}},\label{proof1}
	\end{align}
	where $|\xi_U\ra$ is the state associated to $U$ in the definition of $\epsUPQP$ (Definition 2 in the main text).

	At this point, we appeal to the operational interpretation of the distance induced by the diamond norm given by Theorem \ref{Thdiamond}. In fact, it turns out that \eqref{proof1} can be understood in terms of the optimal probability of distinguishing  the channel $\Tr_{\HH_m} \VV (\cdot \otimes |\xi_U\ra\la\xi_U|)\VV^\dagger$ and the ideal  channel $U ( \cdot ) U^\dagger$, $p^*_{dist}$. We claim  that $$\|\Phi_{\VV}(\sigma)\|_{\BB(\HH_m)}  \ge \sqrt{2} (1 - p^*_{dist})^{1/2}.$$ With this estimate at hand, we can easily finish our proof since, according to Theorem \ref{Thdiamond}, we obtain 
	\beq
	\|\Phi_{\VV}(\sigma)\|_{\BB(\HH_m)}  \ge 
	\left( 1-\frac{1}{2}	\left\| \Tr_{\HH_m} \VV (\cdot \otimes |\xi_U\ra\la\xi_U|)\VV^\dagger -  U (\cdot) U^{\dagger}\right\|_{\diamond}\right)^{\tfrac{1}{2}} \ge \left( 1-\varepsilon	\right)^{\frac{1}{2}}.
	\eeq
	
	To finish our proof, let us show our claim. To this end, according to \eqref{DistProb} applied to the channels $\Tr_{\HH_m} \VV (\cdot \otimes |\xi_U\ra\la\xi_U|)\VV^\dagger$ and $U \, \cdot \, U^\dagger$, we can write
	\begin{align}
		2 (1 - p^*_{dist}) &=  2 - 2\sup_{\rho,Q} \bigg[ \frac{1}{2} \Tr \big[Q \big( \PP \otimes \id_Z  (\rho) \big) \big]  + \frac{1}{2}  \Tr \big[ \big(\id_{\HH_d\otimes\mathcal{Z}}-Q\big) \big((U \otimes \id_\ZZ)  (\rho) U^\dagger \otimes \id_\ZZ \big) \big]  \bigg] \nn\\
		&\le 2 - \Tr \Big[|\gamma\ra\la\gamma|\big(U \otimes \id_\ZZ  |\psi\ra\la\psi| U^\dagger \otimes\id_\ZZ \big) \Big]  - \Tr \Big[ \big(\id_{\HH_d\otimes\mathcal{Z}}-|\gamma\ra\la\gamma|\big)  \big( \PP_{\xi_u} \otimes \id_\mathcal{Z}  (|\psi\ra\la\psi|) \big)  \Big] \nn\\
		&=  2 - 2 + \Tr \Big[|\gamma\ra\la\gamma| \big( \PP_{\xi_U} \otimes \id_\mathcal{Z}  (|\psi\ra\la\psi|) \big) \Big] 
		= \Tr \Big[|\gamma\ra\la\gamma| \big( \PP_{\xi_U} \otimes \id_\mathcal{Z}  (|\psi\ra\la\psi|) \big) \Big] \nn ,
	\end{align}
	where we can recognize \eqref{proof1} in the last expression when $\PP_{\xi_U} (\cdot) = \Tr_{\HH_m}\big[\VV (\: \cdot\: \otimes |\xi_U\ra\la\xi_U|)\VV^\dagger\big]$.
\end{proof}

\subsection{Theorem  \ref{props2}}\label{SecProofs2}

Theorem  \ref{props1}  identifies any $\epsUPQP$, $\mathcal{P}(\cdot) = \mathcal{V} ( \cdot )  \mathcal{V}^{\dagger},\ \mathcal{V}\in \UU(\HH_d \otimes \HH_m)$,  with a $\varepsilon$-embedding which is completely contractive. That is, the previous unitary channel $\mathcal{P}$ defines a map
\beq \begin{array}{cccl}
	\Phi_{\mathcal{V}} : &  \s(\HH_d)& \longrightarrow  &\mathcal{B}(\HH_m) \\
	&	\sigma  	& \mapsto	&	
	\, \Phi_{\mathcal{V}} (\sigma)\, := \Tr_{\HH_d} \mathcal{V} (\sigma^\rmT \otimes \id_{\HH_m}), \end{array}\label{correspondence1}
\eeq which is completely contractive and also a $\varepsilon$-embedding.

Now we show the converse of this correspondence, which is the content of Theorem  \ref{props2}.

\begin{proof}
	As we have already explained, the isometric identification $\nn \BB(\HH_d\otimes\HH_m) \simeq  \mathcal{CB}\left( \s(\HH_d), \BB(\HH_m) \right)$, guarantes that a completely contractive map $\Phi :=\ \s(\HH_d) \longrightarrow  \mathcal{B}(\HH_m)$ must be of the form
	\beq\label{CBform}
	\Phi(\cdot) = \Tr_{\HH_d} \left[  T (\cdot^\rmT \otimes \id_{\HH_m}) \right],\text{ with } \|T\|_{ \BB(\HH_d\otimes\HH_m)} = \| \Phi \|_{\mathcal{CB}\left(\s(\HH_d), \mathcal{B}(\HH_m)\right)} =1.
	\eeq
	
	Hence, by the Russo-Dye Theorem, $T$ can be written as a convex combination of at most $d m$ unitaries: $T = \sum_i \lambda_i \TT_i$, where $\TT_i\in \UU(\HH_d\otimes\HH_m)$ and $\lambda_i\geq 0$ for every $i$, and $\sum_i\lambda_i = 1$. Hence, the mapping $ \Phi$ can be written as:
	$$
	\Phi (\cdot)  =   \Tr_{\HH_d} \Big[  \sum_i \lambda_i \TT_i (\cdot^{\rmT} \otimes \id_{\HH_m})  \Big],
	$$where $\lbrace \TT_i \rbrace_{i=1}^{d m} \subset \UU(\HH_d\otimes\HH_m)$, $\lambda_i \ge 0$ for every $i$ and  $\sum_i \lambda_i = 1$.

	Let us define a new unitary quantum channel $\mathcal{P} (\cdot) = \VV \,\cdot \, \VV^\dagger$ where $$\VV = \sum_i \TT_i \otimes |i\ra\la i|\in  \UU(\HH_d^{\otimes2}\otimes \HH_m^{\otimes2})\qquad \text{ ( recall that }  \TT_i \in \UU(\HH_d\otimes\HH_m) \text{ )}.$$ Now, given an element $\sigma^\rmT = \Tr_\ZZ |\psi\ra\la\gamma| \in\s(\HH_d)$, where $\ZZ$ is an arbitrary finite dimensional Hilbert space, we relate the norm of the image of this element by $\Phi$ with the fidelity between $|\gamma\ra\la \gamma|$ and the state resulting from the action of $\PP\otimes\id_\ZZ$ on $|\psi\ra\la\psi|$. We shorten the notation for this fidelity as:
	\beq\nn
	\mathcal{F}_\PP \big[|\gamma\ra;|\psi\ra\otimes|\xi\ra\big] = \mathcal{F} \bigg[  \Tr_{\HH_d \otimes \HH_m^{\otimes 2}}\Big[\PP \otimes \id_\ZZ  \Big( |\psi\ra\la\psi| \otimes |\xi\ra\la\xi|  \otimes \sum_{i,j=1}^{dm} \sqrt{\lambda_i \lambda_j} | i \ra\la j| \Big)  \Big] ,|\gamma\ra\la\gamma|\bigg] ,
	\eeq
	for any $|\xi\ra \in \mathsf{ball} (\HH_m)$. 
	
	Using the well known fact $\mathcal{F}(\rho, |\gamma\rangle\langle\gamma|)=\big(\Tr[\rho |\gamma\rangle\langle\gamma|]\big)^{\frac{1}{2}}$, we can write:
	\begin{align*} \label{prop2eq1} 
		\mathcal{F}_\PP \big[|\gamma\ra;|\psi\ra\otimes|\xi\ra\big] & =
		\Bigg(\Tr \bigg[ \Big(|\gamma\ra\la\gamma|\otimes \id_{\HH_d \otimes \HH_m^{\otimes 2}}  \Big) \		\big(\PP \otimes \id_\ZZ\big) \Big(|\psi\ra\la\psi| \otimes |\xi\ra\la\xi| \otimes \sum_{i,j} \sqrt{\lambda_i \lambda_j} | i \ra\la j|\Big)\bigg]\Bigg)^{\frac{1}{2}}\nn.\end{align*}
	
	Then, we can replace the second identity on $\HH_m$ by the projector $\sum_{k,l} \sqrt{\lambda_k \lambda_l} | k \ra\la l|$ to conclude that the previous quantity is greater than or equal to
	\begin{align*}
		&\ge \Bigg(\Tr \bigg[ \Big(|\gamma\ra\la\gamma| \otimes \id_{\HH_m} \otimes \sum_{k,l} \sqrt{\lambda_k \lambda_l} | k \ra\la l| \Big) \nn\big(\PP \otimes \id_\ZZ\big) \Big(|\psi\ra\la\psi| \otimes |\xi\ra\la\xi| \otimes \sum_{i,j} \sqrt{\lambda_i \lambda_j} | i \ra\la j|\Big)\bigg] \Bigg)^{\frac{1}{2}} \nn\\&
		= \Bigg( \sum_{i,j} \lambda_i \lambda_j \Tr \bigg[ \big(|\gamma\ra\la\gamma| \otimes  \id_{\HH_m} \big) \nn (\TT_i \otimes \id_\ZZ ) \big(  |\psi\ra\la\psi| \otimes |\xi\ra\la\xi|   \big)  (\TT_j^\dagger \otimes \id_\ZZ)  \bigg] \Bigg)^{\frac{1}{2}} \nn\\&=	\Big\| \sum_i \lambda_i \la \gamma |  \TT_i \otimes \id_\ZZ  |\psi\ra |\xi\ra  \Big\|_{\HH_m}
		\nn\\&= \Big\| \sum_i \lambda_i\Tr_{\HH_d} \big[ \TT_i ( \Tr_\ZZ[|\psi\ra\la\gamma|]  \otimes  \id_{\HH_m} )\big] |\xi\ra  \Big\|_{\HH_m}   =  \big\|  \Phi \big( \sigma^\rmT \big) |\xi\ra  \big\|_{\HH_m}.
	\end{align*}
	Hence, we have shown that 
	\begin{equation}\label{prop2eq1} 
		\mathcal{F}_\PP \big[|\gamma\ra;|\psi\ra\otimes|\xi\ra\big] \geq  \big\|  \Phi \big( \sigma^\rmT \big) |\xi\ra  \big\|_{\HH_m} .
	\end{equation}
	
	At this point, the  relation proven in \eqref{prop2eq1} allows us to connect the previous calculations with the diamond norm, by means of the Fuchs-van de Graaf inequalities \cite[Theorem 1]{FuchsvandeGraaf99}:
	\begin{align}\label{prop2eq2}
		\frac{1}{2}\Big\|  \Tr_{\HH_d \otimes \HH_m^{\otimes 2}} \Big[  \PP  \big(  \cdot \otimes |\xi\ra\la\xi| \otimes \sum_{i,j} \sqrt{\lambda_i \lambda_j} |i \ra\la j| \big) \Big]  - U ( \cdot) U^\dagger )\Big\|_{\diamond} & \le \sup_{\begin{subarray}{c}
		|\psi\ra\in\mathsf{ball}(\HH_d \otimes \ZZ)	 \\ \ZZ	\end{subarray} } \sqrt{1-\mathcal{F}_\PP \big[\id_\ZZ\otimes U|\psi\ra;|\psi\ra\otimes|\xi\ra\big]^2} \nn\\
		&\le \sup_{\begin{subarray}{c}
			|\psi\ra\in\mathsf{ball}(\HH_d \otimes \ZZ)	\\ \ZZ	\end{subarray} } \sqrt{1-\bigg\|  \Phi \Big( \big(\Tr_\ZZ (|\psi\ra\la \psi|) U^\dagger\big)^\rmT \Big) |\xi\ra  \bigg\|_{\HH_m}^2} .\nn
	\end{align}
	
	In order to finish the proof, we need to show that we can choose a unit vector $|\xi_U\ra\in\HH_d$ such that $\|\Phi(\sigma^\rmT) |\xi_U\ra \|_{\HH_m} \ge (1-\delta)$ and, in addition, $|\xi_U\ra$ does not depend on the positive part of $\sigma^\rmT$. Let us be more precise:
	\begin{fact}\label{independence of positive part}
		For every unitary $U\in \BB(\HH_d)$ there exists a unit element $A_U\in \s(\HH_m)$ such that for every positive element $\rho\in \s(\HH_d)$ we have $$\Tr[A_U \Phi(\rho U^\dag)]\geq (1-\delta)\|\rho U^\dag\|_{\s(\HH_d)}=(1-\delta) \|\rho\|_{\s(\HH_d)},$$where $\delta$ is the one in the statement of Theorem  \ref{props2}. 
	\end{fact}
	\begin{proof} (Of Fact \ref{independence of positive part}.)
		
		Indeed, given a unitary $U\in \BB(\HH_d)$, the duality $\s(\HH_d)^*=\BB(\HH_d)$ allows us to identify $U$  with a linear map 
		$\s(\HH_d)^* \simeq \BB(\HH_d)$, 
		$$\begin{array}{cccl} u: & \s(\HH_d) &\longrightarrow& \mathbb{C}\\
		&  \sigma &  \mapsto  &  u(\sigma) := \la U^\rmT , \sigma \ra = \Tr[U \sigma] \end{array}
		$$
		with norm one. Moreover, by looking at the isomorphism  $\Phi =\s(\HH_d) \rightarrow  \Phi(\s(\HH_d))\subset\mathcal{B}(\HH_m)$, we can also define the linear map $a_U:=u\circ \Phi^{-1}: \Phi(\s(\HH_d))\rightarrow \mathbb C$. According to the properties of  $\Phi^{-1}$ it is easy to deduce that $$\|a_U\|\leq \|u\|\|\Phi^{-1}\|\leq \frac{1}{1-\delta}.$$Then, we can invoke Hahn-Banch theorem to conclude the existence of a linear map $\bar{a}_U:\mathcal{B}(\HH_m)\rightarrow \mathbb C$ with $\|\bar{a}_U\|=\|a_U\|$ and such that the map $\bar{a}_U$ is an extension of $a_U$. This can be summarized in the following diagram: 
		
		\[\begin{tikzcd}
		\s(\HH_d)  \arrow[r, "\Phi"]  \arrow[d, "u",swap,start anchor={[yshift=.25ex]}] & \Phi(\s(\HH_d))  \arrow[dl,"a"{name=a_U, below},,start anchor={[xshift=1ex]}] \hspace{-1cm} &  \subset \BB(\HH_m) \arrow[dll,"\bar a_U"{name=ta},yshift=-0.5ex,xshift=0.3ex,start anchor={[xshift=3ex,yshift=.5ex]}]\\
		\mathbb{C}					&  &  & 
		\end{tikzcd}\]
		
		Invoking again the duality $\s(\HH_d)^*=\BB(\HH_d)$ we know that the linear map $\bar{a}_U$ can be identified with an element $B_U\in \s(\HH_m)$ with norm $$\|B_U\|_{\s(\HH_m)}=\|\bar{a}_U:\mathcal{B}(\HH_m)\rightarrow \mathbb C\|\leq \frac{1}{1-\delta}.$$By defining $A_U:=B_U^T/\|B_U\|_{\s(\HH_m)}$ we obtain a unit element in $\s(\HH_m)$ such that
		\begin{align*}\Tr[A_U \Phi(\rho U^\dag)]&=\frac{1}{\|B_U\|_{\s(\HH_m)}}\Tr[B_U^T \Phi(\rho U^\dag)]=\frac{1}{\|B_U\|_{\s(\HH_m)}}\langle B_U, \Phi(\rho U^\dag)\rangle\\&=\frac{1}{\|B_U\|_{\s(\HH_m)}}\bar{a}_U(\Phi(\rho U^\dag))=\frac{1}{\|B_U\|_{\s(\HH_m)}}a_U(\Phi(\rho U^\dag))=\frac{1}{\|B_U\|_{\s(\HH_m)}}u(\rho U^\dag)\\&=\frac{1}{\|B_U\|_{\s(\HH_m)}}\langle U^T,\rho U^\dag\rangle\geq (1-\delta)\Tr[U\rho U^\dag]=(1-\delta)\Tr(\rho)= (1-\delta)\|\rho\|_{\s(\HH_d)}.
		\end{align*}So we have proven Fact \ref{independence of positive part}.
	\end{proof}

	The last obstruction to conclude the proof of Theorem  \ref{props2} is that $A$ is not a pure state (it is not even positive). Then, we must purify it to obtain a unit vector $|\xi_U\ra$ with the desired properties. This can be done by modifying the mapping $\Phi$ and considering the new one:
	$$\begin{array}{cccc}
	\bar \Phi : & \s(\HH_d) & \longrightarrow & \BB(\HH_m^{\otimes 2})\\
	& \sigma  &	\mapsto			&	\Phi(\sigma) \otimes \id_{\HH_m}.
	\end{array}$$

	Indeed, let us assume that $A_U= \sum \mu_i |\alpha_i\ra\la \beta_i |\in \s(\HH_m)$ and  consider the unit vectors $|\xi_U\ra = \sum_i \sqrt{\mu_i} |\alpha_i \ra |i\ra$ and $|\chi_U\ra = \sum_i \sqrt{\mu_i} |\beta_i \ra |i\ra$ in $\HH_m^{\otimes 2}$ so that $\Tr_{\HH_m^{(2nd)}} |\xi_U\ra\la\chi_U| = A_U $.
	Using Fact \ref{independence of positive part} we obtain that 
	\begin{align*}
		\|\rho U^\dagger \|_{\s(\HH_d)}
		& \leq \frac{1}{1-\delta}\Tr[A_U \Phi(\rho U^\dag)]=\frac{1}{1-\delta}\Tr[ |\xi_U\ra\la\chi_U| \bar \Phi (\rho U^\dag)]\leq \frac{1}{1-\delta}\|\bar \Phi (\rho U^\dag)|\xi_U\ra\|_{\HH_m^{\otimes 2}}.\end{align*}
	
	Summing up, we have proven the following
	\begin{fact}\label{fact4}
		For each $U\in\UU(\HH_d)$, there exists a unit vector $| \xi_U \ra \in \HH_m^{\otimes 2}$ such that
		\begin{align*}
			\| \bar \Phi(\rho U^\dagger) |\xi_U\ra \|_{\HH_m^{\otimes 2}} \ge (1-\delta) \| \rho U^\dagger \|_{\s(\HH_d)}, \, \, \text{  for every positive element }\,  \rho \in \s(\HH_d).
		\end{align*}
	\end{fact}
	
	In view of the last point, we have to modify the quantum channel considered at the beginning of the proof to be:
	$$
	\bar \PP (\cdot) = \bar\VV (\cdot ) \bar\VV^\dagger,\  \text{  where}  \, \, \, \, \bar \VV = \VV \otimes \id_{\HH_m} \in \UU(\HH_d^{\otimes 2}\otimes \HH_m^{\otimes3}) .
	$$
	
	Redefining everything accordingly to this, we have
	\beq\nn
	\mathcal{F}_{\bar \PP} \big[|\gamma\ra;|\psi\ra\otimes|\xi\ra\big] = \mathcal{F} \bigg[  \Tr_{\HH_d \otimes \HH_m^{\otimes 3}}\Big[(\bar\PP \otimes \id_\ZZ)  \Big( |\psi\ra\la\psi| \otimes |\xi\ra\la\xi|  \otimes \sum_{i,j} \sqrt{\lambda_i \lambda_j} | i \ra\la j| \Big)  \Big] , |\gamma\ra\la\gamma|\bigg] ,
	\eeq
	and, following the lines of \eqref{prop2eq1}, we obtain:
	\beq
	\mathcal{F}_{\bar \PP} \big[|\gamma\ra;|\psi\ra\otimes|\xi\ra\big] \ge \big\| \bar \Phi (\Tr_\ZZ |\psi\ra\la\gamma|^\rmT) |\xi\ra \big\|_{\HH_m}.
	\eeq
	
	From Fact \ref{fact4} we also know that for $\sigma^\rmT = \Tr_\ZZ\big|\psi\ra\la\gamma|,\text{ with } | \gamma\ra = (U \otimes \id_\ZZ) |\psi\ra$, there exists a vector $|\xi_U \ra \in \mathsf{ball}(\HH_m^{\otimes 2})$ such that
	$$
	\mathcal{F}_{\bar \PP} \big[ U \otimes \id_\ZZ |\psi\ra;|\psi\ra\otimes|\xi_U\ra\big] \ge \Big\| \bar \Phi \Big( \big(\Tr_\ZZ\big[{|\psi\ra\la\psi|\big]U^\dagger}\big) ^\rmT\Big) |\xi_U\ra \Big\|_{\HH_m^{\otimes 2}}\ge (1-\delta)\big\|\Tr_\ZZ\big[|\psi\ra\la\psi|\big]U^\dagger \big\|_{\s(\HH_d)},$$
	for every $|\psi\ra$ in  $\mathsf{ball} (\HH_d\otimes\ZZ),$ and  for every finite dimensional Hilbert space $\ZZ$.
	
	Then, we finally bound
	\begin{align}\nn \hspace{-5mm}
		&\frac{1}{2}\Big\|  \Tr_{\HH_d \otimes \HH_m^{\otimes 3}} \Big[  \bar\PP  \big(  \cdot \otimes |\xi_U\ra\la\xi_U| \otimes \sum_{i,j} \sqrt{\lambda_i \lambda_j} |i \ra\la j| \big) \Big]  - U (\cdot ) U^\dagger \Big\|_{\diamond} \\&\le \sup_{|\psi\ra\in \mathsf{ball}(\HH_d \otimes \ZZ)} \sqrt{1-\mathcal{F}_{\bar \PP} \big[ U \otimes \id_\ZZ |\psi\ra;|\psi\ra\otimes|\xi_U\ra\big]^2} \\ \nn&\le \sqrt{2\delta},
	\end{align}
	
	concluding the proof.
\end{proof}

\subsection{Theorem \ref{ths1}}\label{SecProofs3}

%

In this section, we prove that any $\epsUPQP$, $\PP$, not necessarily unitary,  satisfies the restrictions given by Theorem \ref{ths1}.

\begin{proof}
	First, we consider a Stinespring's dilation for $\PP$:
	$$\VV\in \UU(\HH_d\otimes\HH_{m'}),\ \ \text{such that}\ \  \PP(\cdot) = \Tr_{\HH_{anc}} \VV ( \cdot ) \VV^\dagger,$$
	where $\HH_{m'} = \HH_m \otimes \HH_{anc}$ and $\HH_{anc}$ is an ancillary Hilbert space of dimension equal to the Krauss rank of $\PP$, $rank(\PP) \le (\dim (\HH_d\otimes\HH_m))^2 \equiv (d m)^2$. Fixing the dimension of $\HH_{anc}$, $\VV$ is uniquely determined  up to unitaries on $\HH_{anc}$. 
	
	Now, we construct $\Phi_{\PP}$ as in the unitary case \eqref{correspondence1}:
	\beq\label{correspondence2}
	\begin{array}{cccl}
		\Phi_\PP: & \s(\HH_d) &\lhook\joinrel\longrightarrow& \ \BB(\HH_{m'})\\
		&  \sigma &  \mapsto  &  \Phi_\PP(\sigma) := \Tr_{\HH_d} \VV (\sigma^\rmT \otimes\id_{\HH_{m'}})  ,\end{array}
	\eeq 
	and from Theorem  \ref{props1} we obtain that:
	\beq\label{ineqs}
	\nn\|\sigma\|_{S_1(\HH_d)} \ge \|\Phi_{\PP}(\sigma)\|_{\BB(\HH_m)}  \ge (1-\varepsilon)^{1/2} \|\sigma\|_{S_1(\HH_d)}, \ \text{    for every     }\, \sigma\in S_1(\HH_d).
	\eeq
	
	Using in the first place Proposition  \ref{TypeProp2} and then  Proposition \ref{TypeProp1}, the last inequalities imply that:
	\begin{align} \label{typeconsts2}
		\mathrm{T}_2(S_1(\HH_d))  \le \dfrac{1}{(1-\varepsilon)^{\frac{1}{2}}} \mathrm{T}_2\big( \Phi_{\VV}(S_1(\HH_d)) \big)  \le \frac{1}{(1-\varepsilon)^{\frac{1}{2}}} \mathrm{T}_2\big(\BB(\HH_{m'})\big).
	\end{align}
	
	Taking into account the estimates \eqref{typeconsts} we finally arrive to:
	\begin{equation}
		d  \le \frac{C}{(1-\varepsilon)}\log (\dim \HH_{m'}) \quad  \Rightarrow \quad  \dim \HH_{m'} \ge 2^{\frac{(1-\varepsilon)}{C} d},\nn
	\end{equation}
	where $C$ can be taken equal to $4$ or maybe better.
	Recalling that  $\HH_{m'} = \HH_m \otimes \HH_{anc}$ with $\dim \HH_{anc} \le (d m)^2$ we get the stated bound:
	$$
	m	\ge	\frac{1}{d^{2/3}} 2^{\frac{(1-\varepsilon)d}{3 C}} =2^{ ^{\frac{(1-\varepsilon)d}{3 C} - \frac{2}{3} \log d}}.
	$$
\end{proof}

\begin{observation}
	We notice that the proof remains unchanged if we restrict to the family of elements $\big\lbrace diag(\epsilon_1,\ldots,\epsilon_d)$ $: \epsilon_i \in\lbrace \pm 1\rbrace  \big\rbrace \subset \s(\HH_d)$ instead of considering the action of $\Phi_\PP$ on the whole $\s(\HH_d)$. That is true since these elements are enough to estimate  $  \rmT_2(\HH_d) \ge \sqrt{d} $, just following the lines of the proof of the first part of \eqref{typeconsts}. This means that a programmable processor implementing   the family of unitaries $ \big\lbrace diag(\epsilon_1,\ldots,\epsilon_d)$ $: \epsilon_i \in\lbrace \pm 1\rbrace  \big\rbrace$ up to accuracy $\varepsilon^{-1}$ 
	also has to satisfy the bound of Theorem \ref{ths1}. Explicitly, it means that to program the $2^d$ elements in $\big\lbrace diag(\epsilon_1,\ldots,\epsilon_d)$ $: \epsilon_i \in\lbrace \pm 1\rbrace  \big\rbrace \subset \s(\HH_d)$ a memory of dimension at least $2^{\frac{(1-\varepsilon)}{3C}d - \frac{2}{3} \log d}$ is needed, while a classical memory of dimension $2^d$ is enough to store them with no error.
\end{observation}

%


\begin{thebibliography}{41}%
\makeatletter
\providecommand \@ifxundefined [1]{%
 \@ifx{#1\undefined}
}%
\providecommand \@ifnum [1]{%
 \ifnum #1\expandafter \@firstoftwo
 \else \expandafter \@secondoftwo
 \fi
}%
\providecommand \@ifx [1]{%
 \ifx #1\expandafter \@firstoftwo
 \else \expandafter \@secondoftwo
 \fi
}%
\providecommand \natexlab [1]{#1}%
\providecommand \enquote  [1]{``#1''}%
\providecommand \bibnamefont  [1]{#1}%
\providecommand \bibfnamefont [1]{#1}%
\providecommand \citenamefont [1]{#1}%
\providecommand \href@noop [0]{\@secondoftwo}%
\providecommand \href [0]{\begingroup \@sanitize@url \@href}%
\providecommand \@href[1]{\@@startlink{#1}\@@href}%
\providecommand \@@href[1]{\endgroup#1\@@endlink}%
\providecommand \@sanitize@url [0]{\catcode `\\12\catcode `\$12\catcode
  `\&12\catcode `\#12\catcode `\^12\catcode `\_12\catcode `\%12\relax}%
\providecommand \@@startlink[1]{}%
\providecommand \@@endlink[0]{}%
\providecommand \url  [0]{\begingroup\@sanitize@url \@url }%
\providecommand \@url [1]{\endgroup\@href {#1}{\urlprefix }}%
\providecommand \urlprefix  [0]{URL }%
\providecommand \Eprint [0]{\href }%
\providecommand \doibase [0]{http://dx.doi.org/}%
\providecommand \selectlanguage [0]{\@gobble}%
\providecommand \bibinfo  [0]{\@secondoftwo}%
\providecommand \bibfield  [0]{\@secondoftwo}%
\providecommand \translation [1]{[#1]}%
\providecommand \BibitemOpen [0]{}%
\providecommand \bibitemStop [0]{}%
\providecommand \bibitemNoStop [0]{.\EOS\space}%
\providecommand \EOS [0]{\spacefactor3000\relax}%
\providecommand \BibitemShut  [1]{\csname bibitem#1\endcsname}%
\let\auto@bib@innerbib\@empty
\bibitem [{\citenamefont {Herbert}(1982)}]{Herbert82}%
  \BibitemOpen
  \bibfield  {author} {\bibinfo {author} {\bibfnamefont {N.}~\bibnamefont
  {Herbert}},\ }\href {\doibase 10.1007/BF00729622} {\bibfield  {journal}
  {\bibinfo  {journal} {Foundations of Physics}\ }\textbf {\bibinfo {volume}
  {12}},\ \bibinfo {pages} {1171} (\bibinfo {year} {1982})}\BibitemShut
  {NoStop}%
\bibitem [{\citenamefont {Dieks}(1982)}]{Dieks82}%
  \BibitemOpen
  \bibfield  {author} {\bibinfo {author} {\bibfnamefont {D.}~\bibnamefont
  {Dieks}},\ }\href {\doibase https://doi.org/10.1016/0375-9601(82)90084-6}
  {\bibfield  {journal} {\bibinfo  {journal} {Physics Letters A}\ }\textbf
  {\bibinfo {volume} {92}},\ \bibinfo {pages} {271 } (\bibinfo {year}
  {1982})}\BibitemShut {NoStop}%
\bibitem [{\citenamefont {Milonni}\ and\ \citenamefont
  {Hardies}(1982)}]{Milonni82}%
  \BibitemOpen
  \bibfield  {author} {\bibinfo {author} {\bibfnamefont {P.}~\bibnamefont
  {Milonni}}\ and\ \bibinfo {author} {\bibfnamefont {M.}~\bibnamefont
  {Hardies}},\ }\href {\doibase https://doi.org/10.1016/0375-9601(82)90899-4}
  {\bibfield  {journal} {\bibinfo  {journal} {Physics Letters A}\ }\textbf
  {\bibinfo {volume} {92}},\ \bibinfo {pages} {321 } (\bibinfo {year}
  {1982})}\BibitemShut {NoStop}%
\bibitem [{\citenamefont {W.~K.~Wootters}(1982)}]{WootersZurek82}%
  \BibitemOpen
  \bibfield  {author} {\bibinfo {author} {
  \bibnamefont {W.~K.~Wootters} \and\ \bibfnamefont {W.~H. Zurek}},\ }\href {\doibase doi:10.1038/299802a0}
  {\bibfield  {journal} {\bibinfo  {journal} {Nature}\ }\textbf {\bibinfo
  {volume} {299}},\ \bibinfo {pages} {802 } (\bibinfo {year}
  {1982})}\BibitemShut {NoStop}%
\bibitem [{\citenamefont {Mandel}(1983)}]{Mandel83}%
  \BibitemOpen
  \bibfield  {author} {\bibinfo {author} {\bibfnamefont {L.}~\bibnamefont
  {Mandel}},\ }\href {\doibase http://dx.doi.org/10.1038/304188a0} {\bibfield
  {journal} {\bibinfo  {journal} {Nature}\ }\textbf {\bibinfo {volume} {304}},\
  \bibinfo {pages} {188} (\bibinfo {year} {1983})}\BibitemShut {NoStop}%
\bibitem [{\citenamefont {Pati}\ and\ \citenamefont
  {Braunstein}(2000)}]{PatiBraunstein00}%
  \BibitemOpen
  \bibfield  {author} {\bibinfo {author} {\bibfnamefont {A.~K.}\ \bibnamefont
  {Pati}}\ and\ \bibinfo {author} {\bibfnamefont {S.~L.}\ \bibnamefont
  {Braunstein}},\ }\href {\doibase doi:10.1038/404130b0} {\bibfield  {journal}
  {\bibinfo  {journal} {Nature}\ }\textbf {\bibinfo {volume} {404}},\ \bibinfo
  {pages} {164 } (\bibinfo {year} {2000})}\BibitemShut {NoStop}%
\bibitem [{\citenamefont {Nielsen}\ and\ \citenamefont
  {Chuang}(1997)}]{NielsenChuang97}%
  \BibitemOpen
  \bibfield  {author} {\bibinfo {author} {\bibfnamefont {M.~A.}\ \bibnamefont
  {Nielsen}}\ and\ \bibinfo {author} {\bibfnamefont {I.~L.}\ \bibnamefont
  {Chuang}},\ }\href {\doibase 10.1103/PhysRevLett.79.321} {\bibfield
  {journal} {\bibinfo  {journal} {Phys. Rev. Lett.}\ }\textbf {\bibinfo
  {volume} {79}},\ \bibinfo {pages} {321} (\bibinfo {year} {1997})}\BibitemShut
  {NoStop}%
\bibitem [{Note1()}]{Note1}%
  \BibitemOpen
  \bibinfo {note} {Originally called Programmable Quantum Gate Array \cite
  {NielsenChuang97}.}\BibitemShut {Stop}%
\bibitem [{Note2()}]{Note2}%
  \BibitemOpen
  \bibinfo {note} {Even in this case we could program general quantum channels
  implementing a unitary first and tracing out a part of the
  output.}\BibitemShut {Stop}%
\bibitem [{\citenamefont {Du\ifmmode~\check{s}\else \v{s}\fi{}ek}\ and\
  \citenamefont {Bu\ifmmode~\check{z}\else \v{z}\fi{}ek}(2002)}]{DusekBuzek02}%
  \BibitemOpen
  \bibfield  {author} {\bibinfo {author} {\bibfnamefont {M.}~\bibnamefont
  {Du\ifmmode~\check{s}\else \v{s}\fi{}ek}}\ and\ \bibinfo {author}
  {\bibfnamefont {V.}~\bibnamefont {Bu\ifmmode~\check{z}\else \v{z}\fi{}ek}},\
  }\href {\doibase 10.1103/PhysRevA.66.022112} {\bibfield  {journal} {\bibinfo
  {journal} {Phys. Rev. A}\ }\textbf {\bibinfo {volume} {66}},\ \bibinfo
  {pages} {022112} (\bibinfo {year} {2002})}\BibitemShut {NoStop}%
\bibitem [{\citenamefont {Fiur\'a\ifmmode~\check{s}\else \v{s}\fi{}ek}\ \emph
  {et~al.}(2002)\citenamefont {Fiur\'a\ifmmode~\check{s}\else \v{s}\fi{}ek},
  \citenamefont {Du\ifmmode~\check{s}\else \v{s}\fi{}ek},\ and\ \citenamefont
  {Filip}}]{FiurasekDusekFilip02}%
  \BibitemOpen
  \bibfield  {author} {\bibinfo {author} {\bibfnamefont {J.}~\bibnamefont
  {Fiur\'a\ifmmode~\check{s}\else \v{s}\fi{}ek}}, \bibinfo {author}
  {\bibfnamefont {M.}~\bibnamefont {Du\ifmmode~\check{s}\else \v{s}\fi{}ek}}, \
  and\ \bibinfo {author} {\bibfnamefont {R.}~\bibnamefont {Filip}},\ }\href
  {\doibase 10.1103/PhysRevLett.89.190401} {\bibfield  {journal} {\bibinfo
  {journal} {Phys. Rev. Lett.}\ }\textbf {\bibinfo {volume} {89}},\ \bibinfo
  {pages} {190401} (\bibinfo {year} {2002})}\BibitemShut {NoStop}%
  \bibitem [{\citenamefont {D'Ariano}\ and\ \citenamefont
  {Perinotti}(2005)}]{DarianoPerinnoti05}%
  \BibitemOpen
  \bibfield  {author} {\bibinfo {author} {\bibfnamefont {G.~M.}\ \bibnamefont
  {D'Ariano}}\ and\ \bibinfo {author} {\bibfnamefont {P.}~\bibnamefont
  {Perinotti}},\ }\href {\doibase 10.1103/PhysRevLett.94.090401} {\bibfield
  {journal} {\bibinfo  {journal} {Phys. Rev. Lett.}\ }\textbf {\bibinfo
  {volume} {94}},\ \bibinfo {pages} {090401} (\bibinfo {year}
  {2005})}\BibitemShut {NoStop}%
\bibitem [{\citenamefont {Kim}\ \emph {et~al.}(2001)\citenamefont {Kim},
  \citenamefont {Cheong}, \citenamefont {Lee},\ and\ \citenamefont
  {Lee}}]{KimEtal01}%
  \BibitemOpen
  \bibfield  {author} {\bibinfo {author} {\bibfnamefont {J.}~\bibnamefont
  {Kim}}, \bibinfo {author} {\bibfnamefont {Y.}~\bibnamefont {Cheong}},
  \bibinfo {author} {\bibfnamefont {J.-S.}\ \bibnamefont {Lee}}, \ and\
  \bibinfo {author} {\bibfnamefont {S.}~\bibnamefont {Lee}},\ }\href {\doibase
  10.1103/PhysRevA.65.012302} {\bibfield  {journal} {\bibinfo  {journal} {Phys.
  Rev. A}\ }\textbf {\bibinfo {volume} {65}},\ \bibinfo {pages} {012302}
  (\bibinfo {year} {2001})}\BibitemShut {NoStop}%
\bibitem [{\citenamefont {Hillery}\ \emph {et~al.}(2001)\citenamefont
  {Hillery}, \citenamefont {Bužek},\ and\ \citenamefont
  {Ziman}}]{HilleryBuzekZiman01}%
  \BibitemOpen
  \bibfield  {author} {\bibinfo {author} {\bibfnamefont {M.}~\bibnamefont
  {Hillery}}, \bibinfo {author} {\bibfnamefont {V.}~\bibnamefont {Bužek}}, \
  and\ \bibinfo {author} {\bibfnamefont {M.}~\bibnamefont {Ziman}},\ }\href
  {https://onlinelibrary.wiley.com/doi/abs/10.1002/1521-3978(200110)49:10/11%3C987::AID-PROP987%3E3.0.CO;2-S}
  {\bibfield  {journal} {\bibinfo  {journal} {Fortschritte der Physik}\
  }\textbf {\bibinfo {volume} {49}},\ \bibinfo {pages} {987} (\bibinfo {year}
  {2001})}\BibitemShut {NoStop}%
\bibitem [{\citenamefont {Brazier}\ \emph {et~al.}(2005)\citenamefont
  {Brazier}, \citenamefont {Bu\ifmmode~\check{z}\else \v{z}\fi{}ek},\ and\
  \citenamefont {Knight}}]{BrazierBuzek05}%
  \BibitemOpen
  \bibfield  {author} {\bibinfo {author} {\bibfnamefont {A.}~\bibnamefont
  {Brazier}}, \bibinfo {author} {\bibfnamefont {V.}~\bibnamefont
  {Bu\ifmmode~\check{z}\else \v{z}\fi{}ek}}, \ and\ \bibinfo {author}
  {\bibfnamefont {P.~L.}\ \bibnamefont {Knight}},\ }\href {\doibase
  10.1103/PhysRevA.71.032306} {\bibfield  {journal} {\bibinfo  {journal} {Phys.
  Rev. A}\ }\textbf {\bibinfo {volume} {71}},\ \bibinfo {pages} {032306}
  (\bibinfo {year} {2005})}\BibitemShut {NoStop}%
\bibitem [{\citenamefont {Bera}\ \emph {et~al.}(2009)\citenamefont {Bera},
  \citenamefont {Fenner}, \citenamefont {Green},\ and\ \citenamefont
  {Homer}}]{BeraEtal09}%
  \BibitemOpen
  \bibfield  {author} {\bibinfo {author} {\bibfnamefont {D.}~\bibnamefont
  {Bera}}, \bibinfo {author} {\bibfnamefont {S.}~\bibnamefont {Fenner}},
  \bibinfo {author} {\bibfnamefont {F.}~\bibnamefont {Green}}, \ and\ \bibinfo
  {author} {\bibfnamefont {S.}~\bibnamefont {Homer}},\ }in\ \href@noop {}
  {\emph {\bibinfo {booktitle} {Computing and Combinatorics}}},\ \bibinfo
  {editor} {edited by\ \bibinfo {editor} {\bibfnamefont {H.~Q.}\ \bibnamefont
  {Ngo}}}\ (\bibinfo  {publisher} {Springer Berlin Heidelberg},\ \bibinfo
  {address} {Berlin, Heidelberg},\ \bibinfo {year} {2009})\ pp.\ \bibinfo
  {pages} {418--428}\BibitemShut {NoStop}%
\bibitem [{\citenamefont {Vidal}\ \emph {et~al.}(2002)\citenamefont {Vidal},
  \citenamefont {Masanes},\ and\ \citenamefont {Cirac}}]{VidalCirac02}%
  \BibitemOpen
  \bibfield  {author} {\bibinfo {author} {\bibfnamefont {G.}~\bibnamefont
  {Vidal}}, \bibinfo {author} {\bibfnamefont {L.}~\bibnamefont {Masanes}}, \
  and\ \bibinfo {author} {\bibfnamefont {J.~I.}\ \bibnamefont {Cirac}},\ }\href
  {\doibase 10.1103/PhysRevLett.88.047905} {\bibfield  {journal} {\bibinfo
  {journal} {Phys. Rev. Lett.}\ }\textbf {\bibinfo {volume} {88}},\ \bibinfo
  {pages} {047905} (\bibinfo {year} {2002})}\BibitemShut {NoStop}%
\bibitem [{\citenamefont {Ishizaka}\ and\ \citenamefont
  {Hiroshima}(2008)}]{IshizakaHiroshima08}%
  \BibitemOpen
  \bibfield  {author} {\bibinfo {author} {\bibfnamefont {S.}~\bibnamefont
  {Ishizaka}}\ and\ \bibinfo {author} {\bibfnamefont {T.}~\bibnamefont
  {Hiroshima}},\ }\href {\doibase 10.1103/PhysRevLett.101.240501} {\bibfield
  {journal} {\bibinfo  {journal} {Phys. Rev. Lett.}\ }\textbf {\bibinfo
  {volume} {101}},\ \bibinfo {pages} {240501} (\bibinfo {year}
  {2008})}\BibitemShut {NoStop}%
\bibitem [{\citenamefont {P\'erez-Garc\'{\i}a}(2006)}]{Perez06}%
  \BibitemOpen
  \bibfield  {author} {\bibinfo {author} {\bibfnamefont {D.}~\bibnamefont
  {P\'erez-Garc\'{\i}a}},\ }\href {\doibase 10.1103/PhysRevA.73.052315}
  {\bibfield  {journal} {\bibinfo  {journal} {Phys. Rev. A}\ }\textbf {\bibinfo
  {volume} {73}},\ \bibinfo {pages} {052315} (\bibinfo {year}
  {2006})}\BibitemShut {NoStop}%
\bibitem [{\citenamefont {
z}(2017)}]{Majenz17}%
  \BibitemOpen
  \bibfield  {author} {\bibinfo {author} {\bibfnamefont {C.}~\bibnamefont
  {Majenz}},\ }\emph {\bibinfo {title} {Entropy in quantum information theory -
  Communication and cryptography}},\ \href
  {http://www.math.ku.dk/noter/filer/phd17cm.pdf} {Ph.D. thesis} (\bibinfo
  {year} {2017})\BibitemShut {NoStop}%
 \bibitem [{\citenamefont {Majenz}(2018)}]{Majenz18}%
\BibitemOpen
\bibfield  {author} {\bibinfo {author} {\bibfnamefont {M.}~\bibnamefont
{Chrisstandl}}, \ \bibinfo {author} {\bibfnamefont {F.}~\bibnamefont
{Leditzky}},\ \bibinfo {author} {\bibfnamefont {C.}~\bibnamefont
{Majenz}},\ \bibinfo {author} {\bibfnamefont {G.}~\bibnamefont
{Smith}},\  \bibinfo {author} {\bibfnamefont {F.}~\bibnamefont
{Speelman}}\ and\ \bibinfo {author} {\bibfnamefont {M.}~\bibnamefont
{Walter}},\ } \href {http://arxiv.org/abs/1809.10751} {\bibfield
{journal} {\bibinfo  {journal} {arXiv:1809.10751}\ } (\bibinfo
{year} {2018})}\BibitemShut {NoStop}%
\bibitem [{\citenamefont {Beigi}\ and\ \citenamefont
  {König}(2011)}]{BeigiKonig11}%
  \BibitemOpen
  \bibfield  {author} {\bibinfo {author} {\bibfnamefont {S.}~\bibnamefont
  {Beigi}}\ and\ \bibinfo {author} {\bibfnamefont {R.}~\bibnamefont {König}},\
  }\href {http://stacks.iop.org/1367-2630/13/i=9/a=093036} {\bibfield
  {journal} {\bibinfo  {journal} {New Journal of Physics}\ }\textbf {\bibinfo
  {volume} {13}},\ \bibinfo {pages} {093036} (\bibinfo {year}
  {2011})}\BibitemShut {NoStop}%
\bibitem [{\citenamefont {Aubrun}\ and\ \citenamefont
  {Szarek}(2017)}]{AubrunSzarekBook}%
  \BibitemOpen
  \bibfield  {author} {\bibinfo {author} {\bibfnamefont {G.}~\bibnamefont
  {Aubrun}}\ and\ \bibinfo {author} {\bibfnamefont {S.~J.}\ \bibnamefont
  {Szarek}},\ }\href@noop {} {\emph {\bibinfo {title} {Alice and Bob Meet
  Banach: The Interface of Asymptotic Geometric Analysis and Quantum
  Information Theory}}}\ (\bibinfo  {publisher} {American Mathematical
  Society},\ \bibinfo {year} {2017})\BibitemShut {NoStop}%
\bibitem [{\citenamefont {Palazuelos}\ and\ \citenamefont
  {Vidick}(2016)}]{PalazuelosVidickReview}%
  \BibitemOpen
  \bibfield  {author} {\bibinfo {author} {\bibfnamefont {C.}~\bibnamefont
  {Palazuelos}}\ and\ \bibinfo {author} {\bibfnamefont {T.}~\bibnamefont
  {Vidick}},\ }\href {\doibase 10.1063/1.4938052} {\bibfield  {journal}
  {\bibinfo  {journal} {Journal of Mathematical Physics}\ }\textbf {\bibinfo
  {volume} {57}},\ \bibinfo {pages} {015220} (\bibinfo {year} {2016})},\
  \Eprint {http://arxiv.org/abs/https://doi.org/10.1063/1.4938052}
  {https://doi.org/10.1063/1.4938052} \BibitemShut {NoStop}%
   \bibitem [{Note()}]{SuppMat}%
  \BibitemOpen
  \bibinfo {note} {See Supplemental Material for further details, which
  includes Refs. \cite
  {WatrousNotes1,TomczakJaegBook1,Tomczak741,FuchsvandeGraaf991}}\BibitemShut {NoStop}%
    \bibitem [{\citenamefont {Watrous}(2018)}]{WatrousNotes1}%
  \BibitemOpen
  \bibfield  {author} {\bibinfo {author} {\bibfnamefont {J.}~\bibnamefont
  {Watrous}},\ }\href@noop {} {\emph {\bibinfo {title} {The Theory of Quantum
  Information}}},\ \bibinfo {edition} {1st}\ ed.\ (\bibinfo  {publisher}
  {Cambridge University Press},\ \bibinfo {address} {New York, NY, USA},\
  \bibinfo {year} {2018})\BibitemShut {NoStop}%
  \bibitem [{\citenamefont {Tomczak-Jaegermann}(1989)}]{TomczakJaegBook1}%
  \BibitemOpen
  \bibfield  {author} {\bibinfo {author} {\bibfnamefont {N.}~\bibnamefont
  {Tomczak-Jaegermann}},\ }\href@noop {} {\emph {\bibinfo {title} {Banach-Mazur
  Distances and Finite-Dimensional Operator Ideals}}}\ (\bibinfo  {publisher}
  {Longman Scientific \& Technical},\ \bibinfo {year} {1989})\BibitemShut
  {NoStop}%
  \bibitem [{\citenamefont {Tomczak-Jaegermann}(1974)}]{Tomczak741}%
  \BibitemOpen
  \bibfield  {author} {\bibinfo {author} {\bibfnamefont {N.}~\bibnamefont
  {Tomczak-Jaegermann}},\ }\href {http://eudml.org/doc/217886} {\bibfield
  {journal} {\bibinfo  {journal} {Studia Mathematica}\ }\textbf {\bibinfo
  {volume} {50}},\ \bibinfo {pages} {163} (\bibinfo {year} {1974})}\BibitemShut
  {NoStop}%
  \bibitem [{\citenamefont {Fuchs}\ and\ \citenamefont {van~de
  Graaf}(2006)}]{FuchsvandeGraaf991}%
  \BibitemOpen
  \bibfield  {author} {\bibinfo {author} {\bibfnamefont {C.~A.}\ \bibnamefont
  {Fuchs}}\ and\ \bibinfo {author} {\bibfnamefont {J.}~\bibnamefont {van~de
  Graaf}},\ }\href {\doibase 10.1109/18.761271} {\bibfield  {journal} {\bibinfo
  {journal} {IEEE Trans. Inf. Theor.}\ }\textbf {\bibinfo {volume} {45}},\
  \bibinfo {pages} {1216} (\bibinfo {year} {2006})}\BibitemShut {NoStop}%
\bibitem [{\citenamefont {Hillery}\ \emph {et~al.}(2002)\citenamefont
  {Hillery}, \citenamefont {Bu\ifmmode~\check{z}\else \v{z}\fi{}ek},\ and\
  \citenamefont {Ziman}}]{HilleryBuzekZiman02}%
  \BibitemOpen
  \bibfield  {author} {\bibinfo {author} {\bibfnamefont {M.}~\bibnamefont
  {Hillery}}, \bibinfo {author} {\bibfnamefont {V.}~\bibnamefont
  {Bu\ifmmode~\check{z}\else \v{z}\fi{}ek}}, \ and\ \bibinfo {author}
  {\bibfnamefont {M.}~\bibnamefont {Ziman}},\ }\href {\doibase
  10.1103/PhysRevA.65.022301} {\bibfield  {journal} {\bibinfo  {journal} {Phys.
  Rev. A}\ }\textbf {\bibinfo {volume} {65}},\ \bibinfo {pages} {022301}
  (\bibinfo {year} {2002})}\BibitemShut {NoStop}%
\bibitem [{\citenamefont {Vidal}\ and\ \citenamefont
  {Cirac}(2000)}]{VidalCirac00}%
  \BibitemOpen
  \bibfield  {author} {\bibinfo {author} {\bibfnamefont {G.}~\bibnamefont
  {Vidal}}\ and\ \bibinfo {author} {\bibfnamefont {J.~I.}\ \bibnamefont
  {Cirac}},\ }\href {https://arxiv.org/abs/quant-ph/0012067} {\bibfield
  {journal} {\bibinfo  {journal} {arXiv:quant-ph/0012067}\ } (\bibinfo {year}
  {2000})}\BibitemShut {NoStop}%
\bibitem [{\citenamefont {{Kitaev}}(1997)}]{Kitaev97}%
  \BibitemOpen
  \bibfield  {author} {\bibinfo {author} {\bibfnamefont {A.~Y.}\ \bibnamefont
  {{Kitaev}}},\ }\href {\doibase 10.1070/RM1997v052n06ABEH002155} {\bibfield
  {journal} {\bibinfo  {journal} {Russian Mathematical Surveys}\ }\textbf
  {\bibinfo {volume} {52}},\ \bibinfo {pages} {1191} (\bibinfo {year}
  {1997})}\BibitemShut {NoStop}%
\bibitem [{\citenamefont {Effros}\ and\ \citenamefont {Ruan}(2000)}]{RuanBook}%
  \BibitemOpen
  \bibfield  {author} {\bibinfo {author} {\bibfnamefont {E.~G.}\ \bibnamefont
  {Effros}}\ and\ \bibinfo {author} {\bibfnamefont {Z.-J.}\ \bibnamefont
  {Ruan}},\ }\href@noop {} {\emph {\bibinfo {title} {Operator Spaces}}}\
  (\bibinfo  {publisher} {Oxford University Press},\ \bibinfo {year}
  {2000})\BibitemShut {NoStop}%
  \bibitem [{\citenamefont {Rudin}(1991)}]{Rudin}%
  \BibitemOpen
  \bibfield  {author} {\bibinfo {author} {\bibfnamefont {W.}~\bibnamefont
  {Rudin}},\ }\href@noop {} {\emph {\bibinfo {title} {Functional Analysis}}}\
  (\bibinfo  {publisher} {McGraw-Hill},\ \bibinfo {year} {1991})\BibitemShut
  {NoStop}%
\bibitem [{\citenamefont {Brandao}\ and\ \citenamefont
  {Harrow}(2015)}]{BrandaoHarrow15}%
  \BibitemOpen
  \bibfield  {author} {\bibinfo {author} {\bibfnamefont {F.~G. S.~L.}\
  \bibnamefont {Brandao}}\ and\ \bibinfo {author} {\bibfnamefont {A.~W.}\
  \bibnamefont {Harrow}},\ }\href {\doibase arXiv:1509.05065} {\  arXiv:1509.05065, \  (\bibinfo
  {year} {2015})}\BibitemShut {NoStop}%
\bibitem [{\citenamefont {Sedl\'ak}\ \emph {et~al.}(2007)\citenamefont
  {Sedl\'ak}, \citenamefont {Ziman}, \citenamefont {P\ifmmode~\check{r}\else
  \v{r}\fi{}ibyla}, \citenamefont {Bu\ifmmode~\check{z}\else \v{z}\fi{}ek},\
  and\ \citenamefont {Hillery}}]{SedlakZiman07}%
  \BibitemOpen
  \bibfield  {author} {\bibinfo {author} {\bibfnamefont {M.}~\bibnamefont
  {Sedl\'ak}}, \bibinfo {author} {\bibfnamefont {M.}~\bibnamefont {Ziman}},
  \bibinfo {author} {\bibfnamefont {O.~c.~v.}\ \bibnamefont
  {P\ifmmode~\check{r}\else \v{r}\fi{}ibyla}}, \bibinfo {author} {\bibfnamefont
  {V.}~\bibnamefont {Bu\ifmmode~\check{z}\else \v{z}\fi{}ek}}, \ and\ \bibinfo
  {author} {\bibfnamefont {M.}~\bibnamefont {Hillery}},\ }\href {\doibase
  10.1103/PhysRevA.76.022326} {\bibfield  {journal} {\bibinfo  {journal} {Phys.
  Rev. A}\ }\textbf {\bibinfo {volume} {76}},\ \bibinfo {pages} {022326}
  (\bibinfo {year} {2007})}\BibitemShut {NoStop}%
\bibitem [{\citenamefont {Sent\'{\i}s}\ \emph {et~al.}(2010)\citenamefont
  {Sent\'{\i}s}, \citenamefont {Bagan}, \citenamefont {Calsamiglia},\ and\
  \citenamefont {Mu\~noz Tapia}}]{SentisBagan11}%
  \BibitemOpen
  \bibfield  {author} {\bibinfo {author} {\bibfnamefont {G.}~\bibnamefont
  {Sent\'{\i}s}}, \bibinfo {author} {\bibfnamefont {E.}~\bibnamefont {Bagan}},
  \bibinfo {author} {\bibfnamefont {J.}~\bibnamefont {Calsamiglia}}, \ and\
  \bibinfo {author} {\bibfnamefont {R.}~\bibnamefont {Mu\~noz Tapia}},\ }\href
  {\doibase 10.1103/PhysRevA.82.042312} {\bibfield  {journal} {\bibinfo
  {journal} {Phys. Rev. A}\ }\textbf {\bibinfo {volume} {82}},\ \bibinfo
  {pages} {042312} (\bibinfo {year} {2010})}\BibitemShut {NoStop}%
\bibitem [{\citenamefont {Zhou}\ \emph {et~al.}(2012)\citenamefont {Zhou},
  \citenamefont {Cui}, \citenamefont {Wu},\ and\ \citenamefont
  {Lon}}]{ZhouXinWu12}%
  \BibitemOpen
  \bibfield  {author} {\bibinfo {author} {\bibfnamefont {T.}~\bibnamefont
  {Zhou}}, \bibinfo {author} {\bibfnamefont {J.~X.}\ \bibnamefont {Cui}},
  \bibinfo {author} {\bibfnamefont {X.}~\bibnamefont {Wu}}, \ and\ \bibinfo
  {author} {\bibfnamefont {G.~L.}\ \bibnamefont {Lon}},\ }\href
  {http://dl.acm.org/citation.cfm?id=2481569.2481578} {\bibfield  {journal}
  {\bibinfo  {journal} {Quantum Info. Comput.}\ }\textbf {\bibinfo {volume}
  {12}},\ \bibinfo {pages} {1017} (\bibinfo {year} {2012})}\BibitemShut
  {NoStop}%
\bibitem [{\citenamefont {Zhou}(2012)}]{Zhou12}%
  \BibitemOpen
  \bibfield  {author} {\bibinfo {author} {\bibfnamefont {T.}~\bibnamefont
  {Zhou}},\ }\href {\doibase 10.1007/s11128-011-0327-x} {\bibfield  {journal}
  {\bibinfo  {journal} {Quantum Information Processing}\ }\textbf {\bibinfo
  {volume} {11}},\ \bibinfo {pages} {1669} (\bibinfo {year}
  {2012})}\BibitemShut {NoStop}%
\bibitem [{\citenamefont {Sent\'{\i}s}\ \emph {et~al.}(2013)\citenamefont
  {Sent\'{\i}s}, \citenamefont {Bagan}, \citenamefont {Calsamiglia},\ and\
  \citenamefont {Mu\~noz Tapia}}]{SentisBagan13}%
  \BibitemOpen
  \bibfield  {author} {\bibinfo {author} {\bibfnamefont {G.}~\bibnamefont
  {Sent\'{\i}s}}, \bibinfo {author} {\bibfnamefont {E.}~\bibnamefont {Bagan}},
  \bibinfo {author} {\bibfnamefont {J.}~\bibnamefont {Calsamiglia}}, \ and\
  \bibinfo {author} {\bibfnamefont {R.}~\bibnamefont {Mu\~noz Tapia}},\ }\href
  {\doibase 10.1103/PhysRevA.88.052304} {\bibfield  {journal} {\bibinfo
  {journal} {Phys. Rev. A}\ }\textbf {\bibinfo {volume} {88}},\ \bibinfo
  {pages} {052304} (\bibinfo {year} {2013})}\BibitemShut {NoStop}%
\bibitem [{\citenamefont {Ji}\ \emph {et~al.}(2008)\citenamefont {Ji},
  \citenamefont {Wang}, \citenamefont {Duan}, \citenamefont {Feng},\ and\
  \citenamefont {Ying}}]{JiWangDuanFengYing08}%
  \BibitemOpen
  \bibfield  {author} {\bibinfo {author} {\bibfnamefont {Z.}~\bibnamefont
  {Ji}}, \bibinfo {author} {\bibfnamefont {G.}~\bibnamefont {Wang}}, \bibinfo
  {author} {\bibfnamefont {R.}~\bibnamefont {Duan}}, \bibinfo {author}
  {\bibfnamefont {Y.}~\bibnamefont {Feng}}, \ and\ \bibinfo {author}
  {\bibfnamefont {M.}~\bibnamefont {Ying}},\ }\href {\doibase
  10.1109/TIT.2008.929940} {\bibfield  {journal} {\bibinfo  {journal} {IEEE
  Transactions on Information Theory}\ }\textbf {\bibinfo {volume} {54}},\
  \bibinfo {pages} {5172} (\bibinfo {year} {2008})}\BibitemShut {NoStop}%
\bibitem [{\citenamefont {Suzuki}(2016)}]{Suzuki16}%
  \BibitemOpen
  \bibfield  {author} {\bibinfo {author} {\bibfnamefont {J.}~\bibnamefont
  {Suzuki}},\ }\href {\doibase 10.1103/PhysRevA.94.042306} {\bibfield
  {journal} {\bibinfo  {journal} {Phys. Rev. A}\ }\textbf {\bibinfo {volume}
  {94}},\ \bibinfo {pages} {042306} (\bibinfo {year} {2016})}\BibitemShut
  {NoStop}%
\bibitem [{\citenamefont {Giovannetti}\ \emph {et~al.}(2013)\citenamefont
  {Giovannetti}, \citenamefont {Maccone}, \citenamefont {Morimae},\ and\
  \citenamefont {Rudolph}}]{GiovannettiMaccone13}%
  \BibitemOpen
  \bibfield  {author} {\bibinfo {author} {\bibfnamefont {V.}~\bibnamefont
  {Giovannetti}}, \bibinfo {author} {\bibfnamefont {L.}~\bibnamefont
  {Maccone}}, \bibinfo {author} {\bibfnamefont {T.}~\bibnamefont {Morimae}}, \
  and\ \bibinfo {author} {\bibfnamefont {T.~G.}\ \bibnamefont {Rudolph}},\
  }\href {\doibase 10.1103/PhysRevLett.111.230501} {\bibfield  {journal}
  {\bibinfo  {journal} {Phys. Rev. Lett.}\ }\textbf {\bibinfo {volume} {111}},\
  \bibinfo {pages} {230501} (\bibinfo {year} {2013})}\BibitemShut {NoStop}%
\bibitem [{\citenamefont {P\'erez-Delgado}\ and\ \citenamefont
  {Fitzsimons}(2015)}]{PerezFitzsimons15}%
  \BibitemOpen
  \bibfield  {author} {\bibinfo {author} {\bibfnamefont {C.~A.}\ \bibnamefont
  {P\'erez-Delgado}}\ and\ \bibinfo {author} {\bibfnamefont {J.~F.}\
  \bibnamefont {Fitzsimons}},\ }\href {\doibase 10.1103/PhysRevLett.114.220502}
  {\bibfield  {journal} {\bibinfo  {journal} {Phys. Rev. Lett.}\ }\textbf
  {\bibinfo {volume} {114}},\ \bibinfo {pages} {220502} (\bibinfo {year}
  {2015})}\BibitemShut {NoStop}%
\bibitem [{\citenamefont {Yu}\ \emph {et~al.}(2014)\citenamefont {Yu},
  \citenamefont {P\'erez-Delgado},\ and\ \citenamefont
  {Fitzsimons}}]{PerezFitzsimons14}%
  \BibitemOpen
  \bibfield  {author} {\bibinfo {author} {\bibfnamefont {L.}~\bibnamefont
  {Yu}}, \bibinfo {author} {\bibfnamefont {C.~A.}\ \bibnamefont
  {P\'erez-Delgado}}, \ and\ \bibinfo {author} {\bibfnamefont {J.~F.}\
  \bibnamefont {Fitzsimons}},\ }\href {\doibase 10.1103/PhysRevA.90.050303}
  {\bibfield  {journal} {\bibinfo  {journal} {Phys. Rev. A}\ }\textbf {\bibinfo
  {volume} {90}},\ \bibinfo {pages} {050303} (\bibinfo {year}
  {2014})}\BibitemShut {NoStop}%
\bibitem [{\citenamefont {Bisio}\ \emph {et~al.}(2010)\citenamefont {Bisio},
  \citenamefont {Chiribella}, \citenamefont {D'Ariano}, \citenamefont
  {Facchini},\ and\ \citenamefont {Perinotti}}]{BisioDarianoEtal10}%
  \BibitemOpen
  \bibfield  {author} {\bibinfo {author} {\bibfnamefont {A.}~\bibnamefont
  {Bisio}}, \bibinfo {author} {\bibfnamefont {G.}~\bibnamefont {Chiribella}},
  \bibinfo {author} {\bibfnamefont {G.~M.}\ \bibnamefont {D'Ariano}}, \bibinfo
  {author} {\bibfnamefont {S.}~\bibnamefont {Facchini}}, \ and\ \bibinfo
  {author} {\bibfnamefont {P.}~\bibnamefont {Perinotti}},\ }\href {\doibase
  10.1103/PhysRevA.81.032324} {\bibfield  {journal} {\bibinfo  {journal} {Phys.
  Rev. A}\ }\textbf {\bibinfo {volume} {81}},\ \bibinfo {pages} {032324}
  (\bibinfo {year} {2010}).}\BibitemShut {Stop}%

\end{thebibliography}

\begin{thebibliography}{7}%
\makeatletter
\providecommand \@ifxundefined [1]{%
 \@ifx{#1\undefined}
}%
\providecommand \@ifnum [1]{%
 \ifnum #1\expandafter \@firstoftwo
 \else \expandafter \@secondoftwo
 \fi
}%
\providecommand \@ifx [1]{%
 \ifx #1\expandafter \@firstoftwo
 \else \expandafter \@secondoftwo
 \fi
}%
\providecommand \natexlab [1]{#1}%
\providecommand \enquote  [1]{``#1''}%
\providecommand \bibnamefont  [1]{#1}%
\providecommand \bibfnamefont [1]{#1}%
\providecommand \citenamefont [1]{#1}%
\providecommand \href@noop [0]{\@secondoftwo}%
\providecommand \href [0]{\begingroup \@sanitize@url \@href}%
\providecommand \@href[1]{\@@startlink{#1}\@@href}%
\providecommand \@@href[1]{\endgroup#1\@@endlink}%
\providecommand \@sanitize@url [0]{\catcode `\\12\catcode `\$12\catcode
  `\&12\catcode `\#12\catcode `\^12\catcode `\_12\catcode `\%12\relax}%
\providecommand \@@startlink[1]{}%
\providecommand \@@endlink[0]{}%
\providecommand \url  [0]{\begingroup\@sanitize@url \@url }%
\providecommand \@url [1]{\endgroup\@href {#1}{\urlprefix }}%
\providecommand \urlprefix  [0]{URL }%
\providecommand \Eprint [0]{\href }%
\providecommand \doibase [0]{http://dx.doi.org/}%
\providecommand \selectlanguage [0]{\@gobble}%
\providecommand \bibinfo  [0]{\@secondoftwo}%
\providecommand \bibfield  [0]{\@secondoftwo}%
\providecommand \translation [1]{[#1]}%
\providecommand \BibitemOpen [0]{}%
\providecommand \bibitemStop [0]{}%
\providecommand \bibitemNoStop [0]{.\EOS\space}%
\providecommand \EOS [0]{\spacefactor3000\relax}%
\providecommand \BibitemShut  [1]{\csname bibitem#1\endcsname}%
\let\auto@bib@innerbib\@empty
\bibitem [{\citenamefont {{Kitaev}}(1997)}]{Kitaevv97}%
  \BibitemOpen
  \bibfield  {author} {\bibinfo {author} {\bibfnamefont {A.~Y.}\ \bibnamefont
  {{Kitaev}}},\ }\href {\doibase 10.1070/RM1997v052n06ABEH002155} {\bibfield
  {journal} {\bibinfo  {journal} {Russian Mathematical Surveys}\ }\textbf
  {\bibinfo {volume} {52}},\ \bibinfo {pages} {1191} (\bibinfo {year}
  {1997})}\BibitemShut {NoStop}%
\bibitem [{\citenamefont {Watrous}(2018)}]{WatrousNotes}%
\BibitemOpen
\bibfield  {author} {\bibinfo {author} {\bibfnamefont {J.}~\bibnamefont
{Watrous}},\ }\href@noop {} {\emph {\bibinfo {title} {The Theory of Quantum
Information}}},\ \bibinfo {edition} {1st}\ ed.\ (\bibinfo  {publisher}
{Cambridge University Press},\ \bibinfo {address} {New York, NY, USA},\
\bibinfo {year} {2018})\BibitemShut {NoStop}%
\bibitem [{\citenamefont {Tomczak-Jaegermann}(1989)}]{TomczakJaegBook}%
  \BibitemOpen
  \bibfield  {author} {\bibinfo {author} {\bibfnamefont {N.}~\bibnamefont
  {Tomczak-Jaegermann}},\ }\href@noop {} {\emph {\bibinfo {title} {Banach-Mazur
  Distances and Finite-Dimensional Operator Ideals}}}\ (\bibinfo  {publisher}
  {Longman Scientific \& Technical},\ \bibinfo {year} {1989})\BibitemShut
  {NoStop}%
\bibitem [{\citenamefont {Tomczak-Jaegermann}(1974)}]{Tomczak74}%
  \BibitemOpen
  \bibfield  {author} {\bibinfo {author} {\bibfnamefont {N.}~\bibnamefont
  {Tomczak-Jaegermann}},\ }\href {http://eudml.org/doc/217886} {\bibfield
  {journal} {\bibinfo  {journal} {Studia Mathematica}\ }\textbf {\bibinfo
  {volume} {50}},\ \bibinfo {pages} {163} (\bibinfo {year} {1974})}\BibitemShut
  {NoStop}%
\bibitem [{\citenamefont {Fuchs}\ and\ \citenamefont {van~de
  Graaf}(2006)}]{FuchsvandeGraaf99}%
  \BibitemOpen
  \bibfield  {author} {\bibinfo {author} {\bibfnamefont {C.~A.}\ \bibnamefont
  {Fuchs}}\ and\ \bibinfo {author} {\bibfnamefont {J.}~\bibnamefont {van~de
  Graaf}},\ }\href {\doibase 10.1109/18.761271} {\bibfield  {journal} {\bibinfo
   {journal} {IEEE Trans. Inf. Theor.}\ }\textbf {\bibinfo {volume} {45}},\
  \bibinfo {pages} {1216} (\bibinfo {year} {2006})}\BibitemShut {NoStop}%
\end{thebibliography}
\end{document}